%% file: simrank_lyap_archiv.tex
\documentclass{sig-alternate}
\usepackage[english]{babel}        
\usepackage{graphicx}
\newtheorem{theorem}{Theorem}

\newcommand{\diag}{\mathop{\mathrm{diag}}}
\newcommand{\diagonal}{\mathop{\mathrm{diagonal}}}

\renewcommand{\vec}{\mathop{\mbox{vec}}}
\usepackage{amsmath,amssymb}
\DeclareMathOperator*{\argmin}{arg\,min}
\usepackage{color}
\usepackage{caption}
\usepackage{subcaption}
\usepackage[titletoc,title]{appendix}
\usepackage{pgfplotstable}
\usepackage{pgf}
\usepackage{import}
\usepackage{booktabs}
\usepackage{multirow}
\usepackage[ruled,vlined,linesnumbered]{algorithm2e}
\graphicspath{{./plots/}}

\begin{document}
\title{Linear complexity SimRank computation based on the iterative diagonal estimation} 
\numberofauthors{3}
\author{
\alignauthor I.V. Oseledets\\
       \affaddr{Skolkovo Institute of Science and Technology, Novaya St., 100,  Skolkovo,
143025, Russia \\ Institute of Numerical Mathematics, Russian Academy of Sciences, Gubkina St., 8, Moscow, 119333}\\
       \email{i.oseledets@skoltech.ru}
\alignauthor G.V. Ovchinnikov\\
       \affaddr{Skolkovo Institute of Science and Technology, Novaya St., 100,  Skolkovo,
       143025, Russia \\ Institute for Design Problems in Microelectronics, Russian Academy of Sciences, prosp. 60-letiya Oktyabrya, 9, Moscow, Russia}
       \email{ovgeorge@yandex.ru}
\alignauthor A. M. Katrutsa\\
       \affaddr{Skolkovo Institute of Science and Technology, Novaya St., 100,  Skolkovo,
143025, Russia \\ Moscow Institute of Physics and Technology, Institutskiy Lane 9, Dolgoprudny, 141700, Russia}\\
       \email{aleksandr.katrutsa@phystech.edu}
}


\maketitle

\begin{abstract}
This paper presents a deterministic linear time complexity IDE-SimRank method to approximately compute SimRank with proved error bound. 
SimRank is a well-known similarity measure between graph vertices which relies on graph topology only and is built on intuition that "two objects are similar if they are related to similar objects".
The fixed point equation for direct SimRank computation is the discrete Lyapunov equation with specific diagonal matrix in the right hand side. 
The proposed method is based on estimation of this diagonal matrix with GMRES and use this estimation to compute singe-source and single pairs queries. These computations are executed with the part of series converging to the discrete Lyapunov equation solution.
\end{abstract}


\keywords{SimRank, graph, similarity measure, Lyapunov equation, inexact GMRES}

\section{Introduction}
This paper presents a new method to efficiently compute the SimRank \cite{SimRank} which is a topologically induced similarity measure between two given vertices of a graph. 

Similarity measures for graphs are useful in many applications such as relation mining~\cite{li2013mapreduce}, document-by-document querying~\cite{li2009docsim,Giles2014} and many others.
A major problem in SimRank computation is the high storage and time complexity of the direct iterative process converging to SimRank. Several schemes have been presented for the approximate computation of the SimRank which are based on different concepts~\cite{Lizorkin2010,Li2010,Fogaras2005}. 
In this paper we propose a two-step method for the approximation of the SimRank.  The first and the most expensive step is the computation where we iteratively estimate the diagonal of the SimRank matrix. 
After that the SimRank scores can be computed by the approximate solution of the discrete Lyapunov equation. Such approach has been considered in \cite{Kusumoto2014}. 
The difference is that instead of using Gauss-Seidel method combined with Monte-Carlo computations to estimate the diagonal we use numerical linear algebra techniques. 
We prove that the linear system for the diagonal has bounded condition number, and we have an $\mathcal{O}(n)$ matrix-by-vector product with guaranteed accuracy, thus \emph{inexact GMRES method} is the method of choice. 
Using the theory of IGMRES we get a provable $\mathcal{O}(n)$ complexity algorithm for the computation of SimRank scores. The final SimRank approximation is a sparse matrix. 
We also provide an efficient practical algorithm  for the computation of SimRank.





\section{Problem Statement}
Let $G = (V, E)$ be a graph, where $V$ is a set of vertices and $E$ is a set of edges. The order of the graph is the number of vertices $|V| = n$. Similarity measures between graph vertices are very useful in some applications. One of the approach to define such similarity measure is SimRank~\cite{SimRank}. In the foundation of SimRank definition lies idea that ''two objects are similar if they are referenced by the similar objects''. It is proposed that vertex similarity lies between $0$ and $1$ with vertex being maximally similar to itself with similarity $1$. By $s(a, b)$ denote the SimRank between vertices $a$ and $b$ defined as
\begin{equation}
s(a,b) =
\begin{cases}
1, \quad $if$ \; a=b & \\
0, \quad $if$ \; I(a) = \varnothing \; $or$\; I(b) = \varnothing &\\ 
\displaystyle \frac{c}{|I(a)| |I(b)|}\sum_{i=1}^{|I(a)|}\sum_{j=1}^{|I(b)|}s(I_i(a),I_j(b)), & $otherwise$,
\end{cases}
\label{def::SimRank}
\end{equation}
where $I(v)$ is the set of in-neighbours of vertex $v$, constant $c \in (0,1)$. Denote by $S$ a SimRank matrix, where $(i,j)$-th element is SimRank between $i$-th and $j$-th vertices. To find the SimRank matrix $S$, one writes $s(a, b)$ for all pairs $(a, b)$ of vertices and obtains linear system of $n^2$ equations which has unique solution~\cite{SimRank}. 
Let $A$ be an adjacency matrix of the graph $G$ normalized by columns:
\[
\sum_{i=1}^{n}A_{ij} = 1, \quad j = 1, \ldots, n, 
\]
and
$$ 
  W = \sqrt{c} A.
$$
The SimRank matrix $S$ is the solution of the following equation: 
\begin{equation}
S = W^{\top}SW -\diag(W^{\top}SW) + I = W^{\top}SW + D,
\label{eq::gdl}
\end{equation}
and $D$ is the diagonal matrix $D = -\diag(W^{\top} S W) + I$.
Equation \eqref{eq::gdl} is typically solved by a fixed-point iteration
\begin{equation}
\begin{split}
& S_0 = I,\\
& S_{k+1} = c A^{\top}S_{k}A -c\diag(A^{\top}S_{k}A) + I,
\end{split}
\label{eq::si2}
\end{equation}
where $\diag(P)$ is an operator that maps given matrix $P$ to the diagonal matrix with diagonal entries equal the diagonal entries of the matrix $P$. 
The iteration \eqref{eq::si2} converges if $c < 1$. 
Direct usage of $\eqref{eq::si2}$ requires $\mathcal{O}(n^2)$ memory cells and $\mathcal{O}(n^3)$ operations for one iteration, thus is infeasible for real-world graphs. 
In this paper we propose a new \emph{Iterative Diagonal Estimation method} (IDE-SimRank method) that approximately computes the SimRank in $\mathcal{O}(n)$ time and memory.  

\section{IDE-SimRank method}
In this section we describe IDE-SimRank method and give theoretical foundation which proves the reasoning and accuracy of our method. 

From \eqref{eq::gdl} it is easy to get a linear system with $n$ unknowns. 
Since $\diag{S} = I$ by definition, and 
$$
   S = W^{\top} S W + D, 
$$
where $D$ is a diagonal matrix, we get
\begin{equation}
   I = \diag(W^{\top} S W) + D = F(D), 
   \label{simrank:maineq}
\end{equation}
where $F(D)$ is a linear operator that maps a diagonal matrix (i.e., a vector of length $n$) to a diagonal matrix (also a vector of length $n$), and is defined as
\begin{equation}
   F(D) = D + \diag(W^{\top} S(D) W),
   \label{simrank:F}
\end{equation}
and $S(D)$ is the solution of the discrete Lyapunov equation
$$
   S(D) = W^{\top} S(D) W + D.
$$
So, to evaluate $F(D)$ for a given $D$ we have to solve the discrete Lyapunov equation and take only the diagonal of the solution.
This can not be done exactly in $\mathcal{O}(n)$ complexity, but it is possible to do \emph{approximate computation} of $F(D)$ with guaranteed accuracy. Moreover,the operator $F(D)$ is well-conditioned, so it is natural to use \emph{iterative methods with inexact matrix-by-vector products} to solve \eqref{simrank:maineq}. 
Inexact GMRES \cite{inexact} is typically a method of choice. 
In order to get a working method, we need two components:
\begin{enumerate}
\item Estimates for the condition number of the linear operator $F(D)$. 
\item Algorithm for the computation of $F(D)$ with a given accuracy $\varepsilon$. (and estimate of its complexity). 
\end{enumerate}

Note that the equation~\eqref{simrank:maineq} was used in the paper \cite{Kusumoto2014} under the name \emph{Linearized SimRank}. 
The difference in our approach is that we use sparse matrix arithmetic and inexact iterative method for its solution 
(compared to the Gauss-Seidel method combined with Monte-Carlo estimation to compute $S(D)$).

\subsection{Estimation of condition number of $F(D)$}

Let $\vec(\cdot)$ be an operator that maps an $n \times n$ matrix to a vector of length $n^2$ taking column-by-column. 
Denote by $Pv=\vec(D(v))$ an operator that maps a vector $v$ of length $n$ to a vector of length $n^2$, where $D(v)$ is a diagonal matrix with $v$ on the diagonal.
Now by a slight abuse of notation let $F$ and $S$ act on a vector $d$ of length $n$.
Then, the matrix corresponding to the operator $F(d)$ can be written using Kronecker products as
\begin{equation}\label{simrank:eqformula}
    \begin{split}
        F &= I + P^{\top}(W^{\top} \otimes W^{\top}) (I - W^{\top} \otimes W^{\top})^{-1} P \\
          &= P^{\top}(I - W^{\top} \otimes W^{\top})^{-1} P.
\end{split}
\end{equation}
The matrix $P$ is the submatrix of the $n^2 \times n^2$ identity matrix, thus $F$ is a submatrix of the matrix $(I - W^{\top} \otimes W^{\top})^{-1}$.

The matrix $F$ is the submatrix of the inverse $M$-matrix, thus it is also the inverse $M$-matrix (see~\cite{inv_M}),
i.e. it is non-singular. Moreover, its condition number can be bounded.
\begin{theorem}
\[
    \kappa(F)_1 \leq \frac{2(1+c)}{(1-c)^2}. 
\]
\label{th:cond}
\end{theorem}
\begin{proof}
    For simplicity, introduce the matrix 
    $$
    	Z = W^{\top}\otimes W^{\top}.
    $$
    The matrix $Z$ is nonnegative and $\Vert Z \Vert_1 = c$.
Since $F$ is a submatrix of the matrix $(I - Z)^{-1}$, there exists an $n^2 \times n^2$ permutation matrix $Q$ such that 
$$
Q (I - Z) Q^{\top} = 
\begin{bmatrix} 
A & B \\
C & D \\
\end{bmatrix},
$$
and 
$$
Q (I - Z)^{-1} Q^{\top} = \begin{bmatrix} * & * \\
* & F
\end{bmatrix}.
$$
Using well-known formulas for block matrix inversion, the matrix $F^{-1}$ can be written as the Schur complement
$$
F^{-1} = D - BA^{-1}C,
$$
Consequently, the 1-norm of $F^{-1}$ is bounded in the following way:    
$$
\Vert F^{-1} \Vert_1 \leq \Vert D \Vert_1 + \Vert B \Vert_1 \Vert A^{-1} \Vert_1 \Vert C \Vert_1 
$$
The matrices $D$, $B$, $C$ are submatrices of the matrix $I - Z$, therefore their norms are bounded by $\Vert I - Z \Vert_1 \leq (1 + c)$ (the norms of the submatrices can not exceed the norm of the matrix). 
To estimate $\Vert A^{-1} \Vert_1$ note that $A$ is also a principal submatrix of $(I - Z)$, thus it can be represented as
$$
(I - \widehat{Z}),
$$
where $\Vert \widehat{Z} \Vert_1 \leq c$, therefore using the standard Neumann series argument
$$
  \Vert A^{-1} \Vert_1 \leq \frac{1}{1 - c}.
$$
Finally, 
$$
\Vert F^{-1} \Vert_1 \leq (1 + c) + \frac{(1 + c)^2}{1-c} = \frac{2(1+c)}{1-c}. 
$$
The matrix $F$ is the submatrix of $(I - Z)^{-1}$, thus
$$
   \Vert F \Vert_1 \leq \Vert (I - Z)^{-1} \Vert_1 \leq \frac{1}{1 - c},
$$
and this completes the proof.
\end{proof}

\subsection{Fast approximate matrix-by-vector product}
The key component for the efficient solution of the system~\eqref{simrank:maineq} is the fast evaluation of $F(D)$ for a given $D$. 
The main computational cost comes from the solution of the discrete Lyapunov equation of the form
$$
S = W^{\top} S W + D,
$$
where $D$ is a diagonal matrix. 
The solution can be written as 
\begin{equation}
\begin{split}
    S &= \sum_{k=0}^{\infty} (W^{\top})^k D W^k = \sum_{k=0}^K (W^{\top})^k D W^k + R_{K}=\\ 
      &= S_K + R_K,
\end{split}
\label{simrank:series}
\end{equation}
where $\Vert R_K \Vert_1 \leq c^K$. The truncated series $S_K$ gives an approximation to $S(D)$ with guaranteed accuracy. 

Algorithm~\ref{alg::fast_matvec} presents fast approximate matvec algorithm used further in GMRES. 
In this algorithm the operator $\diagonal(\cdot)$ maps given $n \times n$ matrix to its $n \times 1$ diagonal.
Note, that to get $\mathcal{O}(n)$ complexity we have introduced \emph{thresholding}: the elements smaller than 
$\tau$ are zeroed out, and all computations are implement in sparse matrix arithmetic.
\begin{algorithm}
\SetAlgoLined
\KwData{Scaled adjacency $n \times n$ matrix $W$, given $n \times 1$ vector $x$, threshold $\tau$, number of iterations $K$.}
\KwResult{Approximate result of matvec $y$}
$y_0$ = $x$\\
$X_0 = \diag(x)$\\
\For{$k=1\ldots K$}{
$X_k = W^{\top}X_{k-1}W$\\
$d = \diagonal(X_k)$\\
Threshold to zero all elements of $d$, which absolute values are less than given threshold $\tau$.\\  
$y_{k}$ = $y_{k-1} + d$
}
$y = y_K$
\caption{Fast approximate matvec algorithm}
\label{alg::fast_matvec}
\end{algorithm}

The error of the matrix-by-vector product can be estimated by the following theorem.
\begin{theorem}
    The result of Algorithm~\ref{alg::fast_matvec} satisfies
    \begin{equation}\label{simrank::acc}
        \Vert y - y_K \Vert_{\infty} \leq \tau \frac{(1 + c)^K - 1}{c} + c^K.
    \end{equation}
\end{theorem}
\begin{proof}
    Suppose that $\widehat{X}_k$ is the result of Algorithm~\ref{alg::fast_matvec} for 
    $\tau = 0$ after $k$ steps, and
    $$X_k = \widehat{X}_k + E_k, \quad \Vert E_k \Vert_{C} = \delta_k.$$
    Then,
    $$
       \delta_{k+1} \leq \delta_k + c \delta_k + \tau.
    $$
    It is obvious that 
    $$\delta_{k} \leq \eta_k, $$
    where $\eta_k$ solves  
    $$
    \eta_{k+1} = (1 + c) \eta_k + \tau, 
    $$
    which can be solved as
    $$
       \eta_k = \tau + \frac{( (1 + c)^k - 1)}{c}.
    $$
    The final result is obtained by using a well-known estimate on the remainder of the Neumann series.
\end{proof}
It is easy to get the upper bound on the complexity. 
The number of terms in the SimRank series to get the accuracy $\varepsilon$ can be estimated as $\log_c \varepsilon^{-1}$. 
At each step, the diagonal of the matrix $(W^{\top})^k D W^k$ has to be computed. 
Let $m$ be an average degree of the vertex. 
Then the sparsity of $W^k$ is bounded by $m^k$, and the evaluation reduces to the evaluation of the column norms of the matrix $W^k D^{1/2}$. 
In practice, however, this bound is a significant overestimation.

\subsection{Putting it all together}
The GMRES algorithm is summarized in Algorithm~\ref{alg::GMRES} \cite{gmres}, and it is assumed that the matrix-by-vector product is exact.
All other operations can be easily implemented in $\mathcal{O}(n)$ complexity.
\begin{algorithm}
 \SetAlgoLined
 \KwData{Matrix $A$, right-hand side $b$, initial guess $x_0$, stopping tolerance $\varepsilon$. }
 \KwResult{Approximate solution $x_m$: $\|Ax_m - b\| \leq \varepsilon$}
 
 Start: compute $r_0 = b - Ax_0$, $v_1 = r_0 /\|r_0 \|$, $V_1 = v_1$, $\beta = \|b\|$.
 
 Iterations:
 
 Orthogonalize: $\tilde{v}_{k+1} = Av_k - V_k h_k$, where $h_k = V_k^T Av$.
 
 Normalize: $v_{k+1} = v_{k+1} /\|v_{k+1}\|$.
 
 Update: $V_{k+1} = ( V_k \,\, v_{k+1} )$, $H_k =   
    \left[\begin{array}{cc}
                    H_{k-1}&h_k\\
                    0&\|v_{k+1}\|\\
    \end{array}\right]$,
where the first column in $H_k$ is omitted when $k = 1$.

Solve the least squares problem $y_k = \argmin_{y}\| \beta e_1 - H_k y\|$.

$x_m = x_0 + V_m y_m$.

Restart: compute $\|r_m \|= \|b - Ax_m\|$. Stop if $\|r_m \| \leq \varepsilon$.
Otherwise: $x_0 = x_m , v_1 = r_m /\|r_m \|$ and start Iterations again.
\caption{GMRES algorithm for the solution of the linear system}\label{alg::GMRES}
\end{algorithm}

If the matrix-by-vector products are inexact, the following Theorem gives the error bound.
\begin{theorem}\cite{inexact}
\label{th:error}
After $m$ steps of the inexact GMRES procedure, the following estimation for norm of approximate and real residues holds:
\[
\|r_m - \tilde{r}_m\| \leq \varepsilon
\] 
if for any $i \leq m$ 
\[
\|\tilde{E}_i\| \leq \frac{\sigma_{m}(H_m)}{m\|\tilde{r}_m\|}\varepsilon,
\]
where $\sigma_m(H_m)$ is a minimal singular value of the Hessenberg matrix corresponding to GMRES process and $\tilde{E}_i$ is an error corresponding to matvec on the $i$-th iteration.
\end{theorem}

\section{Comparison with existing methods}
SimRank algorithm computes similarities between vertices of the input graph $G = (V, E)$. 
Here we compute a single-source SimRank and a one-pair SimRank. 
The single-source SimRank is the vector with SimRank scores between given vertex $a \in V$ and all other vertices $b \in V$. 
The one-pair SimRank is the similarity measure $s(a, b)$ between two given vertices $a$ and $b$. 

The proposed method has memory requirement $\mathcal{O}(n)$ and computational complexity $\mathcal{O}(n)$ of the pre-computation step. 
Moreover, we compute the sparse approximation to the full SimRank matrix.

For the readers convenience computational and memory complexities of the previously proposed methods are presented Table~\ref{tab::tab::single-pair_comparison} and~\ref{tab::single-source_comparison}. 
\begin{table}[!ht]
\caption{Complexities of the single-pair SimRank algorithms}
\begin{tabular}{|c|c|c|c|}
\hline 
\multirow{2}{*}{Paper} & \multicolumn{2}{|c|}{Time} & \multirow{2}{*}{Memory}\\ 
\cline{2-3}
 & Query & Precomputation & \\
\hline
\cite{Fogaras2005} & $\mathcal{O}(kN)$ & $\mathcal{O}((N+d)n)$ & $\mathcal{O}(nN)$\\ 
\cite{Li2010} & $\mathcal{O}(r^4)$ & $\mathcal{O}(r^4n^2)$ & $\mathcal{O}(n^2r^2 + r^4)$\\
\cite{He2014} & $\mathcal{O}(k|E|^2)$ & Not needed & $\mathcal{O}(n^2 )$\\
This & $\mathcal{O}(n)$ & $\mathcal{O}(n)$ & $\mathcal{O}(n)$\\ 
\hline 
\end{tabular}
\label{tab::tab::single-pair_comparison}
\end{table}

\begin{table}[!ht]
\centering
\caption{Complexities of the single-source SimRank algorithms}
\begin{tabular}{|c|c|c|c|}
\hline
\multirow{2}{*}{Paper} & \multicolumn{2}{|c|}{Time} & \multirow{2}{*}{Memory}\\ 
\cline{2-3}
 & Query & Precomputation & \\
\hline
\cite{Kusumoto2014} & $\mathcal{O}(Rk|S|)$ & $\mathcal{O}(nk(R+PQ))$ & $\mathcal{O}(m + nP)$\\
\cline{2-3}
\cite{Lee2012} & \multicolumn{2}{|c|}{$\mathcal{O}(kd^k)$}& $\mathcal{O}(d^k)$\\
\cline{2-3}
This & $\mathcal{O}(1)$ & $\mathcal{O}(1)$ & $\mathcal{O}(1)$\\ 
\hline
\end{tabular}
\label{tab::single-source_comparison}
\end{table}

Some papers like \cite{Onizuka2013}, \cite{Li2010}, \cite{Yu2013}, \cite{Yu2015} consider a solution of the exact Sylvester equation without any estimation as SimRank approximation. 
However, this approach has two fundamental problems. 
The first problem is that the exact solution of discrete Lyapunov equation is only an approximation to the initial SimRank definition, so this solution leads to additional errors. 
We treat this problem by estimation of the diagonal item in discrete Lyapunov equation and get the correct approximate discrete Lyapunov equation for SimRank. 
The second problem is to solve Sylvester equation as proposed in \cite{Li2010} one needs invert adjacency matrix of the graph which is unstable and leads to loss of sparsity. Instead of invert the matrix of linear operator we use iterative method GMRES with fast approximate matvec implementation.


\section{Numerical experiments}
\subsection{Synthetic test}
To confirm the $\mathcal{O}(n)$ complexity of the proposed method, we generate random adjacency matrices with fixed number of nonzero elements in every column and compute SimRank for corresponding graphs. 
The dependence of time on $n$  is shown on Figure~\ref{fig::time_vs_n}.
The other parameters are threshold $\tau = 10^{-4}$, the scale parameter $c = 0.6$ and number of iteration $K = 50$. 
Here $nnz$ is a number of non-zero elements in every column.  
The computational cost increases when the adjacency matrix becomes more dense, which is natural.
\begin{figure}[!ht]
\centering
\resizebox{!}{0.4\textwidth}{
\import{./plots/}{time_vs_n.pgf}
}
\caption{Dependence of time to solve linear system on $n$ for randomly generated graphs with different number of non-zeros}
\label{fig::time_vs_n}
\end{figure}
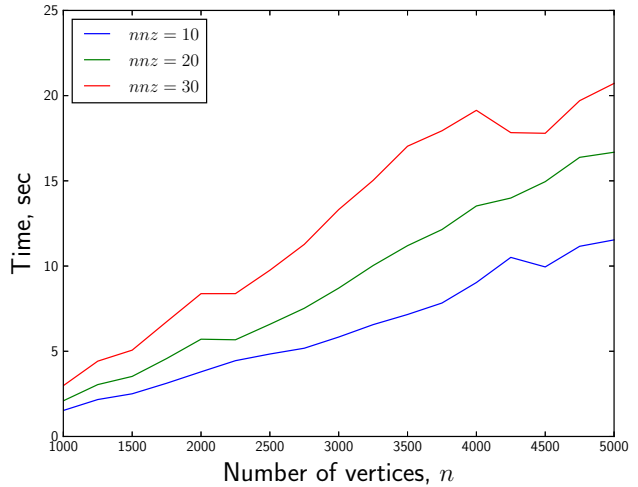

The similar plot for dependence of memory to solve the linear system on the number of vertices $n$ is presented in Figure~\ref{fig::mem_vs_n}. 
The other parameters are the same: threshold $\tau = 10^{-4}$, the scale parameter $c = 0.6$ and number of iteration $K = 50$. 
Here $nnz$ is a number of non-zero elements in every column.
The plot shows that the required memory linearly or sub-linearly depends on the number of vertices in the graph. 
But if the adjacency matrix of the graph is enough sparse, then the required memory is constant and does not depend on the number of vertices.   

\begin{figure}[!ht]
\centering
\resizebox{!}{0.4\textwidth}{
\import{./plots/}{mem_vs_n.pgf}
}
\caption{Dependence of memory to solve linear system on $n$ for randomly generated graphs with different number of non-zeros}
\label{fig::mem_vs_n}
\end{figure}
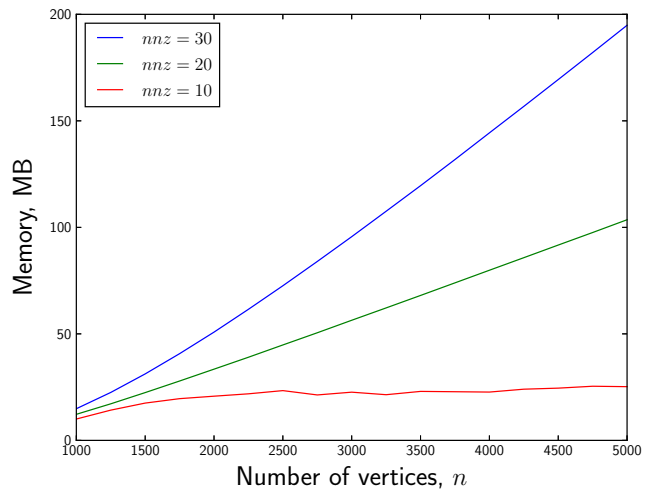
\subsection{DIMACS10 collection}
In this section we experimentally study the approximation accuracy of our method. The experiments are carried out on graphs from DIMACS10 Challenge Collection\footnote{https://www.cise.ufl.edu/research/sparse/matrices/DIMACS10/}. 
The list of the considered graphs is presented in Table~\ref{tab::graphs}. 

\begin{table}[!ht]
\caption{NDCG accuracy for single-source query to the graphs from DIMACS10 Challenge Collection}
\centering
\begin{tabular}{|p{1.9cm}cccc|}\hline 
Name & $n$ & $nnz$ & $nnz/n$ & $1 -$ NDCG@$n$\\ 
\hline 
chesapeake & 39 &  340 &  8.72  & $1.3 \cdot 10^{-9}$  \\ 
data & 2851 &  30186 &  10.59 & $3.1 \cdot 10^{-8}$\\ 
delaunay n10 & 1024 & 6112 &  5.97 & $2\cdot10^{-9}$\\ 
delaunay n11 & 2048 & 12254 & 5.98 & $1.2 \cdot 10^{-8}$\\ 
delaunay n12 & 4096 & 24528 & 5.99 & $4 \cdot 10^{-8}$\\ 
delaunay n13 & 8192 & 49094 & 5.99 & $2.5\cdot10^{-7}$\\ 
uk & 4824 & 13674 & 2.83 & $2.74\cdot10^{-7}$\\ 
vsp data and seymourl & 9167 & 111732 & 12.19 & $10^{-8}$\\ 
\hline 
\end{tabular}
\label{tab::graphs}
\end{table}

Table~\ref{tab::graphs} presents $1-$NDCG@$n$ measure for convenience. To compute the NDCG@$n$ measure we make $q = 100$ random queries to SimRank and SimRank approximation for every graph except chesapeake graph $(q = 39)$. After that we have two vectors $s$ and $\tilde{s}$ of correct SimRank scores and approximate SimRank scores between query and all other graph vertices. The NDCG measure \cite{NDCG} is defined by the following equation:
\[
\text{NDCG} = \frac{1}{Z}\sum_{i = 1}^n \frac{2^{rel_i} - 1}{\log_2(i + 1)},
\]
where $i$ is an index of concept, according to the sorted approximate SimRank scores $\tilde{s}$, $rel_i = s_i$ is the ground-truth SimRank score between the query and the $i$-th vertex, and $Z$ is a normalization constant. The other parameters are $c = 0.6$, $K = 50$ and $\tau = 10^{-3}$. Also $nnz$ is the total number of non-zero elements in the adjacency matrix, $nnz/n$ is the average degree of vertex. 
\begin{figure}[!ht]
\centering
\resizebox{!}{0.4\textwidth}{
\import{./plots/}{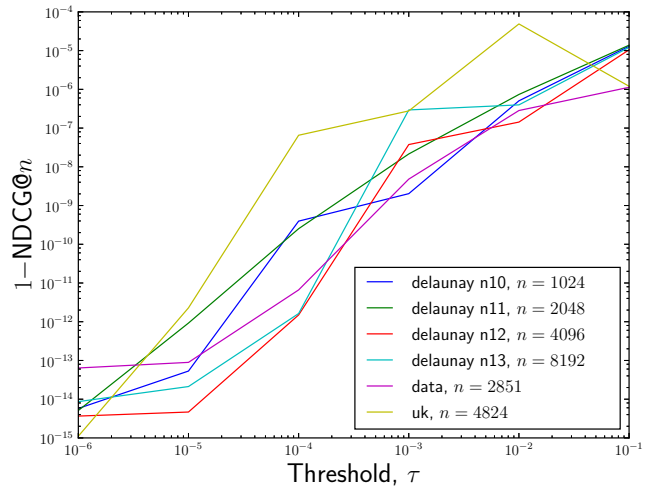}
}
\caption{Dependence NDCG on the thresholds for graph from DIMACS10 collection}
\end{figure}

\begin{figure}[!ht]
\centering
\resizebox{!}{0.4\textwidth}{
\import{./plots/}{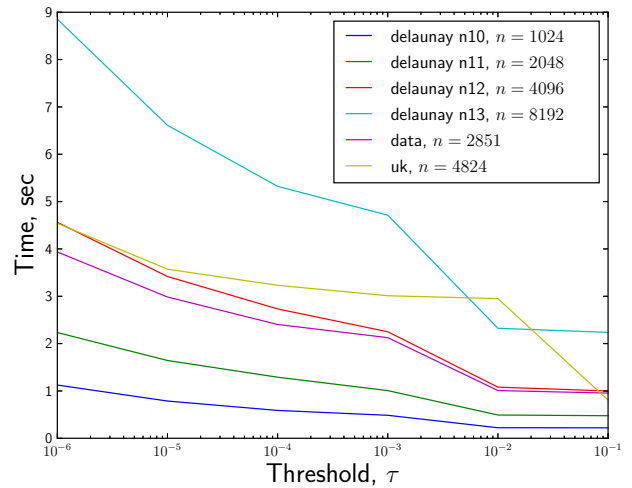}
}
\caption{Dependence of time to solve linear system on the thresholds for graph from DIMACS10 collection}
\end{figure}

\subsection{Experiment with Wikipedia}
We used Simple English Wikipedia corpus to find semantic relatedness between concepts. The undirected graph corresponding to Simple Wikipedia corpus has $150495$ vertices and $4454023$ edges. 
The direct SimRank matrix computation of such large graph is infeasible. Therefore, we assess the quality of our SimRank approximation method not with approximation error but with rationality of the obtained similar concepts. We use the following parameters in the experiment: $c = 0.6$, number of iteration $k = 10$, $\tau=10^{-4}$. 

Table~\ref{tab::wiki} shows some examples of similar concepts extracted from Simple English Wikipedia corpus by the proposed SimRank approximation algorithm. 
Each column hasthe queried concept in the top and the most similar concepts to the queried one in the other rows. 
Every column is sorted according to SimRank scores given by SimRank matrix approximation. 
We do not display these scores because of the space limitation: the scores differ in 4-th or 5-th significant figures.
\begin{figure}[!ht]
\centering
\resizebox{!}{0.4\textwidth}{
\import{./plots/}{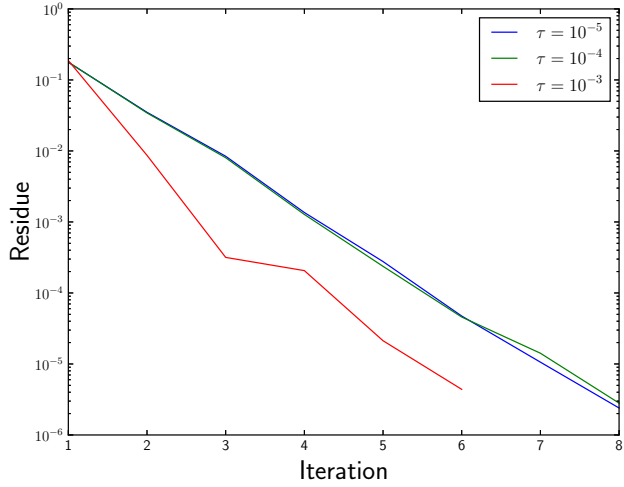}
}
\caption{Convergence of the GMRES with different thresholds $\tau$}
\label{fig::gmres_converg}
\end{figure}
\begin{table}[!ht]
\centering
\caption{Similar concepts according to proposed SimRank approximation algorithm}
\begin{tabular}{|p{3cm}|p{2cm}|p{3cm}|}
\hline
\emph{GNU} & \emph{Earth} & \emph{Liquid}\\
\hline
Richard Matthew Stallman & South Pole-Aitken basin & Plasma (matter)\\
\hline
Linux operating system & Frame of referance & Matters \\
\hline
Hurd & Interplanetary internet & Particle theory of matter\\
\hline
Debian linux & Supernova 1987A & Hematological\\
\hline
Linux (kernel) & Probotector & Blude\\
\hline
*nix & Near Earth Object & Human blood\\
\hline
\end{tabular}
\label{tab::wiki}
\end{table}

\section{Conclusions and future work}
An important research direction is the study of hypergraphs, when the adjacency matrix is replaced by the adjacency tensor. 
We plan to investigate this issue. Also, the computation of the SimRank by summing of the Neumann series can be improved by using more advanced iterative method, like the IGMRES approach considered here, but it requires a lot of technical work. 
\section{Acknowledgements}
The authors thank D. Kolesnikov for fruitful discussions.

\bibliographystyle{abbrv}
\bibliography{references}

\end{document}

%% file: time_vs_n.pgf
\begingroup%
\makeatletter%
\begin{pgfpicture}%
\pgfpathrectangle{\pgfpointorigin}{\pgfqpoint{8.000000in}{6.000000in}}%
\pgfusepath{use as bounding box, clip}%
\begin{pgfscope}%
\pgfsetbuttcap%
\pgfsetmiterjoin%
\definecolor{currentfill}{rgb}{1.000000,1.000000,1.000000}%
\pgfsetfillcolor{currentfill}%
\pgfsetlinewidth{0.000000pt}%
\definecolor{currentstroke}{rgb}{1.000000,1.000000,1.000000}%
\pgfsetstrokecolor{currentstroke}%
\pgfsetdash{}{0pt}%
\pgfpathmoveto{\pgfqpoint{0.000000in}{0.000000in}}%
\pgfpathlineto{\pgfqpoint{8.000000in}{0.000000in}}%
\pgfpathlineto{\pgfqpoint{8.000000in}{6.000000in}}%
\pgfpathlineto{\pgfqpoint{0.000000in}{6.000000in}}%
\pgfpathclose%
\pgfusepath{fill}%
\end{pgfscope}%
\begin{pgfscope}%
\pgfsetbuttcap%
\pgfsetmiterjoin%
\definecolor{currentfill}{rgb}{1.000000,1.000000,1.000000}%
\pgfsetfillcolor{currentfill}%
\pgfsetlinewidth{0.000000pt}%
\definecolor{currentstroke}{rgb}{0.000000,0.000000,0.000000}%
\pgfsetstrokecolor{currentstroke}%
\pgfsetstrokeopacity{0.000000}%
\pgfsetdash{}{0pt}%
\pgfpathmoveto{\pgfqpoint{1.000000in}{0.600000in}}%
\pgfpathlineto{\pgfqpoint{7.200000in}{0.600000in}}%
\pgfpathlineto{\pgfqpoint{7.200000in}{5.400000in}}%
\pgfpathlineto{\pgfqpoint{1.000000in}{5.400000in}}%
\pgfpathclose%
\pgfusepath{fill}%
\end{pgfscope}%
\begin{pgfscope}%
\pgfpathrectangle{\pgfqpoint{1.000000in}{0.600000in}}{\pgfqpoint{6.200000in}{4.800000in}} %
\pgfusepath{clip}%
\pgfsetrectcap%
\pgfsetroundjoin%
\pgfsetlinewidth{1.003750pt}%
\definecolor{currentstroke}{rgb}{0.000000,0.000000,1.000000}%
\pgfsetstrokecolor{currentstroke}%
\pgfsetdash{}{0pt}%
\pgfpathmoveto{\pgfqpoint{1.000000in}{0.892917in}}%
\pgfpathlineto{\pgfqpoint{1.387500in}{1.016010in}}%
\pgfpathlineto{\pgfqpoint{1.775000in}{1.080876in}}%
\pgfpathlineto{\pgfqpoint{2.162500in}{1.199509in}}%
\pgfpathlineto{\pgfqpoint{2.550000in}{1.328162in}}%
\pgfpathlineto{\pgfqpoint{2.937500in}{1.453905in}}%
\pgfpathlineto{\pgfqpoint{3.325000in}{1.528916in}}%
\pgfpathlineto{\pgfqpoint{3.712500in}{1.594249in}}%
\pgfpathlineto{\pgfqpoint{4.100000in}{1.720271in}}%
\pgfpathlineto{\pgfqpoint{4.487500in}{1.860007in}}%
\pgfpathlineto{\pgfqpoint{4.875000in}{1.974239in}}%
\pgfpathlineto{\pgfqpoint{5.262500in}{2.103468in}}%
\pgfpathlineto{\pgfqpoint{5.650000in}{2.334190in}}%
\pgfpathlineto{\pgfqpoint{6.037500in}{2.617586in}}%
\pgfpathlineto{\pgfqpoint{6.425000in}{2.510143in}}%
\pgfpathlineto{\pgfqpoint{6.812500in}{2.742554in}}%
\pgfpathlineto{\pgfqpoint{7.200000in}{2.814767in}}%
\pgfusepath{stroke}%
\end{pgfscope}%
\begin{pgfscope}%
\pgfpathrectangle{\pgfqpoint{1.000000in}{0.600000in}}{\pgfqpoint{6.200000in}{4.800000in}} %
\pgfusepath{clip}%
\pgfsetrectcap%
\pgfsetroundjoin%
\pgfsetlinewidth{1.003750pt}%
\definecolor{currentstroke}{rgb}{0.000000,0.500000,0.000000}%
\pgfsetstrokecolor{currentstroke}%
\pgfsetdash{}{0pt}%
\pgfpathmoveto{\pgfqpoint{1.000000in}{1.001808in}}%
\pgfpathlineto{\pgfqpoint{1.387500in}{1.184323in}}%
\pgfpathlineto{\pgfqpoint{1.775000in}{1.276724in}}%
\pgfpathlineto{\pgfqpoint{2.162500in}{1.477662in}}%
\pgfpathlineto{\pgfqpoint{2.550000in}{1.695863in}}%
\pgfpathlineto{\pgfqpoint{2.937500in}{1.689266in}}%
\pgfpathlineto{\pgfqpoint{3.325000in}{1.863790in}}%
\pgfpathlineto{\pgfqpoint{3.712500in}{2.044309in}}%
\pgfpathlineto{\pgfqpoint{4.100000in}{2.271312in}}%
\pgfpathlineto{\pgfqpoint{4.487500in}{2.525757in}}%
\pgfpathlineto{\pgfqpoint{4.875000in}{2.749551in}}%
\pgfpathlineto{\pgfqpoint{5.262500in}{2.931322in}}%
\pgfpathlineto{\pgfqpoint{5.650000in}{3.196803in}}%
\pgfpathlineto{\pgfqpoint{6.037500in}{3.286384in}}%
\pgfpathlineto{\pgfqpoint{6.425000in}{3.470828in}}%
\pgfpathlineto{\pgfqpoint{6.812500in}{3.744741in}}%
\pgfpathlineto{\pgfqpoint{7.200000in}{3.802330in}}%
\pgfusepath{stroke}%
\end{pgfscope}%
\begin{pgfscope}%
\pgfpathrectangle{\pgfqpoint{1.000000in}{0.600000in}}{\pgfqpoint{6.200000in}{4.800000in}} %
\pgfusepath{clip}%
\pgfsetrectcap%
\pgfsetroundjoin%
\pgfsetlinewidth{1.003750pt}%
\definecolor{currentstroke}{rgb}{1.000000,0.000000,0.000000}%
\pgfsetstrokecolor{currentstroke}%
\pgfsetdash{}{0pt}%
\pgfpathmoveto{\pgfqpoint{1.000000in}{1.172460in}}%
\pgfpathlineto{\pgfqpoint{1.387500in}{1.448485in}}%
\pgfpathlineto{\pgfqpoint{1.775000in}{1.572462in}}%
\pgfpathlineto{\pgfqpoint{2.162500in}{1.892445in}}%
\pgfpathlineto{\pgfqpoint{2.550000in}{2.208900in}}%
\pgfpathlineto{\pgfqpoint{2.937500in}{2.209145in}}%
\pgfpathlineto{\pgfqpoint{3.325000in}{2.472656in}}%
\pgfpathlineto{\pgfqpoint{3.712500in}{2.764318in}}%
\pgfpathlineto{\pgfqpoint{4.100000in}{3.156313in}}%
\pgfpathlineto{\pgfqpoint{4.487500in}{3.485197in}}%
\pgfpathlineto{\pgfqpoint{4.875000in}{3.870206in}}%
\pgfpathlineto{\pgfqpoint{5.262500in}{4.043961in}}%
\pgfpathlineto{\pgfqpoint{5.650000in}{4.274394in}}%
\pgfpathlineto{\pgfqpoint{6.037500in}{4.023146in}}%
\pgfpathlineto{\pgfqpoint{6.425000in}{4.015688in}}%
\pgfpathlineto{\pgfqpoint{6.812500in}{4.383034in}}%
\pgfpathlineto{\pgfqpoint{7.200000in}{4.577831in}}%
\pgfusepath{stroke}%
\end{pgfscope}%
\begin{pgfscope}%
\pgfsetrectcap%
\pgfsetmiterjoin%
\pgfsetlinewidth{1.003750pt}%
\definecolor{currentstroke}{rgb}{0.000000,0.000000,0.000000}%
\pgfsetstrokecolor{currentstroke}%
\pgfsetdash{}{0pt}%
\pgfpathmoveto{\pgfqpoint{1.000000in}{5.400000in}}%
\pgfpathlineto{\pgfqpoint{7.200000in}{5.400000in}}%
\pgfusepath{stroke}%
\end{pgfscope}%
\begin{pgfscope}%
\pgfsetrectcap%
\pgfsetmiterjoin%
\pgfsetlinewidth{1.003750pt}%
\definecolor{currentstroke}{rgb}{0.000000,0.000000,0.000000}%
\pgfsetstrokecolor{currentstroke}%
\pgfsetdash{}{0pt}%
\pgfpathmoveto{\pgfqpoint{7.200000in}{0.600000in}}%
\pgfpathlineto{\pgfqpoint{7.200000in}{5.400000in}}%
\pgfusepath{stroke}%
\end{pgfscope}%
\begin{pgfscope}%
\pgfsetrectcap%
\pgfsetmiterjoin%
\pgfsetlinewidth{1.003750pt}%
\definecolor{currentstroke}{rgb}{0.000000,0.000000,0.000000}%
\pgfsetstrokecolor{currentstroke}%
\pgfsetdash{}{0pt}%
\pgfpathmoveto{\pgfqpoint{1.000000in}{0.600000in}}%
\pgfpathlineto{\pgfqpoint{7.200000in}{0.600000in}}%
\pgfusepath{stroke}%
\end{pgfscope}%
\begin{pgfscope}%
\pgfsetrectcap%
\pgfsetmiterjoin%
\pgfsetlinewidth{1.003750pt}%
\definecolor{currentstroke}{rgb}{0.000000,0.000000,0.000000}%
\pgfsetstrokecolor{currentstroke}%
\pgfsetdash{}{0pt}%
\pgfpathmoveto{\pgfqpoint{1.000000in}{0.600000in}}%
\pgfpathlineto{\pgfqpoint{1.000000in}{5.400000in}}%
\pgfusepath{stroke}%
\end{pgfscope}%
\begin{pgfscope}%
\pgfsetbuttcap%
\pgfsetroundjoin%
\definecolor{currentfill}{rgb}{0.000000,0.000000,0.000000}%
\pgfsetfillcolor{currentfill}%
\pgfsetlinewidth{0.501875pt}%
\definecolor{currentstroke}{rgb}{0.000000,0.000000,0.000000}%
\pgfsetstrokecolor{currentstroke}%
\pgfsetdash{}{0pt}%
\pgfsys@defobject{currentmarker}{\pgfqpoint{0.000000in}{0.000000in}}{\pgfqpoint{0.000000in}{0.055556in}}{%
\pgfpathmoveto{\pgfqpoint{0.000000in}{0.000000in}}%
\pgfpathlineto{\pgfqpoint{0.000000in}{0.055556in}}%
\pgfusepath{stroke,fill}%
}%
\begin{pgfscope}%
\pgfsys@transformshift{1.000000in}{0.600000in}%
\pgfsys@useobject{currentmarker}{}%
\end{pgfscope}%
\end{pgfscope}%
\begin{pgfscope}%
\pgfsetbuttcap%
\pgfsetroundjoin%
\definecolor{currentfill}{rgb}{0.000000,0.000000,0.000000}%
\pgfsetfillcolor{currentfill}%
\pgfsetlinewidth{0.501875pt}%
\definecolor{currentstroke}{rgb}{0.000000,0.000000,0.000000}%
\pgfsetstrokecolor{currentstroke}%
\pgfsetdash{}{0pt}%
\pgfsys@defobject{currentmarker}{\pgfqpoint{0.000000in}{-0.055556in}}{\pgfqpoint{0.000000in}{0.000000in}}{%
\pgfpathmoveto{\pgfqpoint{0.000000in}{0.000000in}}%
\pgfpathlineto{\pgfqpoint{0.000000in}{-0.055556in}}%
\pgfusepath{stroke,fill}%
}%
\begin{pgfscope}%
\pgfsys@transformshift{1.000000in}{5.400000in}%
\pgfsys@useobject{currentmarker}{}%
\end{pgfscope}%
\end{pgfscope}%
\begin{pgfscope}%
\pgftext[x=1.000000in,y=0.544444in,,top]{{\sffamily\fontsize{12.000000}{14.400000}\selectfont 1000}}%
\end{pgfscope}%
\begin{pgfscope}%
\pgfsetbuttcap%
\pgfsetroundjoin%
\definecolor{currentfill}{rgb}{0.000000,0.000000,0.000000}%
\pgfsetfillcolor{currentfill}%
\pgfsetlinewidth{0.501875pt}%
\definecolor{currentstroke}{rgb}{0.000000,0.000000,0.000000}%
\pgfsetstrokecolor{currentstroke}%
\pgfsetdash{}{0pt}%
\pgfsys@defobject{currentmarker}{\pgfqpoint{0.000000in}{0.000000in}}{\pgfqpoint{0.000000in}{0.055556in}}{%
\pgfpathmoveto{\pgfqpoint{0.000000in}{0.000000in}}%
\pgfpathlineto{\pgfqpoint{0.000000in}{0.055556in}}%
\pgfusepath{stroke,fill}%
}%
\begin{pgfscope}%
\pgfsys@transformshift{1.775000in}{0.600000in}%
\pgfsys@useobject{currentmarker}{}%
\end{pgfscope}%
\end{pgfscope}%
\begin{pgfscope}%
\pgfsetbuttcap%
\pgfsetroundjoin%
\definecolor{currentfill}{rgb}{0.000000,0.000000,0.000000}%
\pgfsetfillcolor{currentfill}%
\pgfsetlinewidth{0.501875pt}%
\definecolor{currentstroke}{rgb}{0.000000,0.000000,0.000000}%
\pgfsetstrokecolor{currentstroke}%
\pgfsetdash{}{0pt}%
\pgfsys@defobject{currentmarker}{\pgfqpoint{0.000000in}{-0.055556in}}{\pgfqpoint{0.000000in}{0.000000in}}{%
\pgfpathmoveto{\pgfqpoint{0.000000in}{0.000000in}}%
\pgfpathlineto{\pgfqpoint{0.000000in}{-0.055556in}}%
\pgfusepath{stroke,fill}%
}%
\begin{pgfscope}%
\pgfsys@transformshift{1.775000in}{5.400000in}%
\pgfsys@useobject{currentmarker}{}%
\end{pgfscope}%
\end{pgfscope}%
\begin{pgfscope}%
\pgftext[x=1.775000in,y=0.544444in,,top]{{\sffamily\fontsize{12.000000}{14.400000}\selectfont 1500}}%
\end{pgfscope}%
\begin{pgfscope}%
\pgfsetbuttcap%
\pgfsetroundjoin%
\definecolor{currentfill}{rgb}{0.000000,0.000000,0.000000}%
\pgfsetfillcolor{currentfill}%
\pgfsetlinewidth{0.501875pt}%
\definecolor{currentstroke}{rgb}{0.000000,0.000000,0.000000}%
\pgfsetstrokecolor{currentstroke}%
\pgfsetdash{}{0pt}%
\pgfsys@defobject{currentmarker}{\pgfqpoint{0.000000in}{0.000000in}}{\pgfqpoint{0.000000in}{0.055556in}}{%
\pgfpathmoveto{\pgfqpoint{0.000000in}{0.000000in}}%
\pgfpathlineto{\pgfqpoint{0.000000in}{0.055556in}}%
\pgfusepath{stroke,fill}%
}%
\begin{pgfscope}%
\pgfsys@transformshift{2.550000in}{0.600000in}%
\pgfsys@useobject{currentmarker}{}%
\end{pgfscope}%
\end{pgfscope}%
\begin{pgfscope}%
\pgfsetbuttcap%
\pgfsetroundjoin%
\definecolor{currentfill}{rgb}{0.000000,0.000000,0.000000}%
\pgfsetfillcolor{currentfill}%
\pgfsetlinewidth{0.501875pt}%
\definecolor{currentstroke}{rgb}{0.000000,0.000000,0.000000}%
\pgfsetstrokecolor{currentstroke}%
\pgfsetdash{}{0pt}%
\pgfsys@defobject{currentmarker}{\pgfqpoint{0.000000in}{-0.055556in}}{\pgfqpoint{0.000000in}{0.000000in}}{%
\pgfpathmoveto{\pgfqpoint{0.000000in}{0.000000in}}%
\pgfpathlineto{\pgfqpoint{0.000000in}{-0.055556in}}%
\pgfusepath{stroke,fill}%
}%
\begin{pgfscope}%
\pgfsys@transformshift{2.550000in}{5.400000in}%
\pgfsys@useobject{currentmarker}{}%
\end{pgfscope}%
\end{pgfscope}%
\begin{pgfscope}%
\pgftext[x=2.550000in,y=0.544444in,,top]{{\sffamily\fontsize{12.000000}{14.400000}\selectfont 2000}}%
\end{pgfscope}%
\begin{pgfscope}%
\pgfsetbuttcap%
\pgfsetroundjoin%
\definecolor{currentfill}{rgb}{0.000000,0.000000,0.000000}%
\pgfsetfillcolor{currentfill}%
\pgfsetlinewidth{0.501875pt}%
\definecolor{currentstroke}{rgb}{0.000000,0.000000,0.000000}%
\pgfsetstrokecolor{currentstroke}%
\pgfsetdash{}{0pt}%
\pgfsys@defobject{currentmarker}{\pgfqpoint{0.000000in}{0.000000in}}{\pgfqpoint{0.000000in}{0.055556in}}{%
\pgfpathmoveto{\pgfqpoint{0.000000in}{0.000000in}}%
\pgfpathlineto{\pgfqpoint{0.000000in}{0.055556in}}%
\pgfusepath{stroke,fill}%
}%
\begin{pgfscope}%
\pgfsys@transformshift{3.325000in}{0.600000in}%
\pgfsys@useobject{currentmarker}{}%
\end{pgfscope}%
\end{pgfscope}%
\begin{pgfscope}%
\pgfsetbuttcap%
\pgfsetroundjoin%
\definecolor{currentfill}{rgb}{0.000000,0.000000,0.000000}%
\pgfsetfillcolor{currentfill}%
\pgfsetlinewidth{0.501875pt}%
\definecolor{currentstroke}{rgb}{0.000000,0.000000,0.000000}%
\pgfsetstrokecolor{currentstroke}%
\pgfsetdash{}{0pt}%
\pgfsys@defobject{currentmarker}{\pgfqpoint{0.000000in}{-0.055556in}}{\pgfqpoint{0.000000in}{0.000000in}}{%
\pgfpathmoveto{\pgfqpoint{0.000000in}{0.000000in}}%
\pgfpathlineto{\pgfqpoint{0.000000in}{-0.055556in}}%
\pgfusepath{stroke,fill}%
}%
\begin{pgfscope}%
\pgfsys@transformshift{3.325000in}{5.400000in}%
\pgfsys@useobject{currentmarker}{}%
\end{pgfscope}%
\end{pgfscope}%
\begin{pgfscope}%
\pgftext[x=3.325000in,y=0.544444in,,top]{{\sffamily\fontsize{12.000000}{14.400000}\selectfont 2500}}%
\end{pgfscope}%
\begin{pgfscope}%
\pgfsetbuttcap%
\pgfsetroundjoin%
\definecolor{currentfill}{rgb}{0.000000,0.000000,0.000000}%
\pgfsetfillcolor{currentfill}%
\pgfsetlinewidth{0.501875pt}%
\definecolor{currentstroke}{rgb}{0.000000,0.000000,0.000000}%
\pgfsetstrokecolor{currentstroke}%
\pgfsetdash{}{0pt}%
\pgfsys@defobject{currentmarker}{\pgfqpoint{0.000000in}{0.000000in}}{\pgfqpoint{0.000000in}{0.055556in}}{%
\pgfpathmoveto{\pgfqpoint{0.000000in}{0.000000in}}%
\pgfpathlineto{\pgfqpoint{0.000000in}{0.055556in}}%
\pgfusepath{stroke,fill}%
}%
\begin{pgfscope}%
\pgfsys@transformshift{4.100000in}{0.600000in}%
\pgfsys@useobject{currentmarker}{}%
\end{pgfscope}%
\end{pgfscope}%
\begin{pgfscope}%
\pgfsetbuttcap%
\pgfsetroundjoin%
\definecolor{currentfill}{rgb}{0.000000,0.000000,0.000000}%
\pgfsetfillcolor{currentfill}%
\pgfsetlinewidth{0.501875pt}%
\definecolor{currentstroke}{rgb}{0.000000,0.000000,0.000000}%
\pgfsetstrokecolor{currentstroke}%
\pgfsetdash{}{0pt}%
\pgfsys@defobject{currentmarker}{\pgfqpoint{0.000000in}{-0.055556in}}{\pgfqpoint{0.000000in}{0.000000in}}{%
\pgfpathmoveto{\pgfqpoint{0.000000in}{0.000000in}}%
\pgfpathlineto{\pgfqpoint{0.000000in}{-0.055556in}}%
\pgfusepath{stroke,fill}%
}%
\begin{pgfscope}%
\pgfsys@transformshift{4.100000in}{5.400000in}%
\pgfsys@useobject{currentmarker}{}%
\end{pgfscope}%
\end{pgfscope}%
\begin{pgfscope}%
\pgftext[x=4.100000in,y=0.544444in,,top]{{\sffamily\fontsize{12.000000}{14.400000}\selectfont 3000}}%
\end{pgfscope}%
\begin{pgfscope}%
\pgfsetbuttcap%
\pgfsetroundjoin%
\definecolor{currentfill}{rgb}{0.000000,0.000000,0.000000}%
\pgfsetfillcolor{currentfill}%
\pgfsetlinewidth{0.501875pt}%
\definecolor{currentstroke}{rgb}{0.000000,0.000000,0.000000}%
\pgfsetstrokecolor{currentstroke}%
\pgfsetdash{}{0pt}%
\pgfsys@defobject{currentmarker}{\pgfqpoint{0.000000in}{0.000000in}}{\pgfqpoint{0.000000in}{0.055556in}}{%
\pgfpathmoveto{\pgfqpoint{0.000000in}{0.000000in}}%
\pgfpathlineto{\pgfqpoint{0.000000in}{0.055556in}}%
\pgfusepath{stroke,fill}%
}%
\begin{pgfscope}%
\pgfsys@transformshift{4.875000in}{0.600000in}%
\pgfsys@useobject{currentmarker}{}%
\end{pgfscope}%
\end{pgfscope}%
\begin{pgfscope}%
\pgfsetbuttcap%
\pgfsetroundjoin%
\definecolor{currentfill}{rgb}{0.000000,0.000000,0.000000}%
\pgfsetfillcolor{currentfill}%
\pgfsetlinewidth{0.501875pt}%
\definecolor{currentstroke}{rgb}{0.000000,0.000000,0.000000}%
\pgfsetstrokecolor{currentstroke}%
\pgfsetdash{}{0pt}%
\pgfsys@defobject{currentmarker}{\pgfqpoint{0.000000in}{-0.055556in}}{\pgfqpoint{0.000000in}{0.000000in}}{%
\pgfpathmoveto{\pgfqpoint{0.000000in}{0.000000in}}%
\pgfpathlineto{\pgfqpoint{0.000000in}{-0.055556in}}%
\pgfusepath{stroke,fill}%
}%
\begin{pgfscope}%
\pgfsys@transformshift{4.875000in}{5.400000in}%
\pgfsys@useobject{currentmarker}{}%
\end{pgfscope}%
\end{pgfscope}%
\begin{pgfscope}%
\pgftext[x=4.875000in,y=0.544444in,,top]{{\sffamily\fontsize{12.000000}{14.400000}\selectfont 3500}}%
\end{pgfscope}%
\begin{pgfscope}%
\pgfsetbuttcap%
\pgfsetroundjoin%
\definecolor{currentfill}{rgb}{0.000000,0.000000,0.000000}%
\pgfsetfillcolor{currentfill}%
\pgfsetlinewidth{0.501875pt}%
\definecolor{currentstroke}{rgb}{0.000000,0.000000,0.000000}%
\pgfsetstrokecolor{currentstroke}%
\pgfsetdash{}{0pt}%
\pgfsys@defobject{currentmarker}{\pgfqpoint{0.000000in}{0.000000in}}{\pgfqpoint{0.000000in}{0.055556in}}{%
\pgfpathmoveto{\pgfqpoint{0.000000in}{0.000000in}}%
\pgfpathlineto{\pgfqpoint{0.000000in}{0.055556in}}%
\pgfusepath{stroke,fill}%
}%
\begin{pgfscope}%
\pgfsys@transformshift{5.650000in}{0.600000in}%
\pgfsys@useobject{currentmarker}{}%
\end{pgfscope}%
\end{pgfscope}%
\begin{pgfscope}%
\pgfsetbuttcap%
\pgfsetroundjoin%
\definecolor{currentfill}{rgb}{0.000000,0.000000,0.000000}%
\pgfsetfillcolor{currentfill}%
\pgfsetlinewidth{0.501875pt}%
\definecolor{currentstroke}{rgb}{0.000000,0.000000,0.000000}%
\pgfsetstrokecolor{currentstroke}%
\pgfsetdash{}{0pt}%
\pgfsys@defobject{currentmarker}{\pgfqpoint{0.000000in}{-0.055556in}}{\pgfqpoint{0.000000in}{0.000000in}}{%
\pgfpathmoveto{\pgfqpoint{0.000000in}{0.000000in}}%
\pgfpathlineto{\pgfqpoint{0.000000in}{-0.055556in}}%
\pgfusepath{stroke,fill}%
}%
\begin{pgfscope}%
\pgfsys@transformshift{5.650000in}{5.400000in}%
\pgfsys@useobject{currentmarker}{}%
\end{pgfscope}%
\end{pgfscope}%
\begin{pgfscope}%
\pgftext[x=5.650000in,y=0.544444in,,top]{{\sffamily\fontsize{12.000000}{14.400000}\selectfont 4000}}%
\end{pgfscope}%
\begin{pgfscope}%
\pgfsetbuttcap%
\pgfsetroundjoin%
\definecolor{currentfill}{rgb}{0.000000,0.000000,0.000000}%
\pgfsetfillcolor{currentfill}%
\pgfsetlinewidth{0.501875pt}%
\definecolor{currentstroke}{rgb}{0.000000,0.000000,0.000000}%
\pgfsetstrokecolor{currentstroke}%
\pgfsetdash{}{0pt}%
\pgfsys@defobject{currentmarker}{\pgfqpoint{0.000000in}{0.000000in}}{\pgfqpoint{0.000000in}{0.055556in}}{%
\pgfpathmoveto{\pgfqpoint{0.000000in}{0.000000in}}%
\pgfpathlineto{\pgfqpoint{0.000000in}{0.055556in}}%
\pgfusepath{stroke,fill}%
}%
\begin{pgfscope}%
\pgfsys@transformshift{6.425000in}{0.600000in}%
\pgfsys@useobject{currentmarker}{}%
\end{pgfscope}%
\end{pgfscope}%
\begin{pgfscope}%
\pgfsetbuttcap%
\pgfsetroundjoin%
\definecolor{currentfill}{rgb}{0.000000,0.000000,0.000000}%
\pgfsetfillcolor{currentfill}%
\pgfsetlinewidth{0.501875pt}%
\definecolor{currentstroke}{rgb}{0.000000,0.000000,0.000000}%
\pgfsetstrokecolor{currentstroke}%
\pgfsetdash{}{0pt}%
\pgfsys@defobject{currentmarker}{\pgfqpoint{0.000000in}{-0.055556in}}{\pgfqpoint{0.000000in}{0.000000in}}{%
\pgfpathmoveto{\pgfqpoint{0.000000in}{0.000000in}}%
\pgfpathlineto{\pgfqpoint{0.000000in}{-0.055556in}}%
\pgfusepath{stroke,fill}%
}%
\begin{pgfscope}%
\pgfsys@transformshift{6.425000in}{5.400000in}%
\pgfsys@useobject{currentmarker}{}%
\end{pgfscope}%
\end{pgfscope}%
\begin{pgfscope}%
\pgftext[x=6.425000in,y=0.544444in,,top]{{\sffamily\fontsize{12.000000}{14.400000}\selectfont 4500}}%
\end{pgfscope}%
\begin{pgfscope}%
\pgfsetbuttcap%
\pgfsetroundjoin%
\definecolor{currentfill}{rgb}{0.000000,0.000000,0.000000}%
\pgfsetfillcolor{currentfill}%
\pgfsetlinewidth{0.501875pt}%
\definecolor{currentstroke}{rgb}{0.000000,0.000000,0.000000}%
\pgfsetstrokecolor{currentstroke}%
\pgfsetdash{}{0pt}%
\pgfsys@defobject{currentmarker}{\pgfqpoint{0.000000in}{0.000000in}}{\pgfqpoint{0.000000in}{0.055556in}}{%
\pgfpathmoveto{\pgfqpoint{0.000000in}{0.000000in}}%
\pgfpathlineto{\pgfqpoint{0.000000in}{0.055556in}}%
\pgfusepath{stroke,fill}%
}%
\begin{pgfscope}%
\pgfsys@transformshift{7.200000in}{0.600000in}%
\pgfsys@useobject{currentmarker}{}%
\end{pgfscope}%
\end{pgfscope}%
\begin{pgfscope}%
\pgfsetbuttcap%
\pgfsetroundjoin%
\definecolor{currentfill}{rgb}{0.000000,0.000000,0.000000}%
\pgfsetfillcolor{currentfill}%
\pgfsetlinewidth{0.501875pt}%
\definecolor{currentstroke}{rgb}{0.000000,0.000000,0.000000}%
\pgfsetstrokecolor{currentstroke}%
\pgfsetdash{}{0pt}%
\pgfsys@defobject{currentmarker}{\pgfqpoint{0.000000in}{-0.055556in}}{\pgfqpoint{0.000000in}{0.000000in}}{%
\pgfpathmoveto{\pgfqpoint{0.000000in}{0.000000in}}%
\pgfpathlineto{\pgfqpoint{0.000000in}{-0.055556in}}%
\pgfusepath{stroke,fill}%
}%
\begin{pgfscope}%
\pgfsys@transformshift{7.200000in}{5.400000in}%
\pgfsys@useobject{currentmarker}{}%
\end{pgfscope}%
\end{pgfscope}%
\begin{pgfscope}%
\pgftext[x=7.200000in,y=0.544444in,,top]{{\sffamily\fontsize{12.000000}{14.400000}\selectfont 5000}}%
\end{pgfscope}%
\begin{pgfscope}%
\pgftext[x=4.100000in,y=0.313705in,,top]{{\sffamily\fontsize{22.000000}{26.400000}\selectfont Number of vertices, \(\displaystyle n\)}}%
\end{pgfscope}%
\begin{pgfscope}%
\pgfsetbuttcap%
\pgfsetroundjoin%
\definecolor{currentfill}{rgb}{0.000000,0.000000,0.000000}%
\pgfsetfillcolor{currentfill}%
\pgfsetlinewidth{0.501875pt}%
\definecolor{currentstroke}{rgb}{0.000000,0.000000,0.000000}%
\pgfsetstrokecolor{currentstroke}%
\pgfsetdash{}{0pt}%
\pgfsys@defobject{currentmarker}{\pgfqpoint{0.000000in}{0.000000in}}{\pgfqpoint{0.055556in}{0.000000in}}{%
\pgfpathmoveto{\pgfqpoint{0.000000in}{0.000000in}}%
\pgfpathlineto{\pgfqpoint{0.055556in}{0.000000in}}%
\pgfusepath{stroke,fill}%
}%
\begin{pgfscope}%
\pgfsys@transformshift{1.000000in}{0.600000in}%
\pgfsys@useobject{currentmarker}{}%
\end{pgfscope}%
\end{pgfscope}%
\begin{pgfscope}%
\pgfsetbuttcap%
\pgfsetroundjoin%
\definecolor{currentfill}{rgb}{0.000000,0.000000,0.000000}%
\pgfsetfillcolor{currentfill}%
\pgfsetlinewidth{0.501875pt}%
\definecolor{currentstroke}{rgb}{0.000000,0.000000,0.000000}%
\pgfsetstrokecolor{currentstroke}%
\pgfsetdash{}{0pt}%
\pgfsys@defobject{currentmarker}{\pgfqpoint{-0.055556in}{0.000000in}}{\pgfqpoint{0.000000in}{0.000000in}}{%
\pgfpathmoveto{\pgfqpoint{0.000000in}{0.000000in}}%
\pgfpathlineto{\pgfqpoint{-0.055556in}{0.000000in}}%
\pgfusepath{stroke,fill}%
}%
\begin{pgfscope}%
\pgfsys@transformshift{7.200000in}{0.600000in}%
\pgfsys@useobject{currentmarker}{}%
\end{pgfscope}%
\end{pgfscope}%
\begin{pgfscope}%
\pgftext[x=0.944444in,y=0.600000in,right,]{{\sffamily\fontsize{12.000000}{14.400000}\selectfont 0}}%
\end{pgfscope}%
\begin{pgfscope}%
\pgfsetbuttcap%
\pgfsetroundjoin%
\definecolor{currentfill}{rgb}{0.000000,0.000000,0.000000}%
\pgfsetfillcolor{currentfill}%
\pgfsetlinewidth{0.501875pt}%
\definecolor{currentstroke}{rgb}{0.000000,0.000000,0.000000}%
\pgfsetstrokecolor{currentstroke}%
\pgfsetdash{}{0pt}%
\pgfsys@defobject{currentmarker}{\pgfqpoint{0.000000in}{0.000000in}}{\pgfqpoint{0.055556in}{0.000000in}}{%
\pgfpathmoveto{\pgfqpoint{0.000000in}{0.000000in}}%
\pgfpathlineto{\pgfqpoint{0.055556in}{0.000000in}}%
\pgfusepath{stroke,fill}%
}%
\begin{pgfscope}%
\pgfsys@transformshift{1.000000in}{1.560000in}%
\pgfsys@useobject{currentmarker}{}%
\end{pgfscope}%
\end{pgfscope}%
\begin{pgfscope}%
\pgfsetbuttcap%
\pgfsetroundjoin%
\definecolor{currentfill}{rgb}{0.000000,0.000000,0.000000}%
\pgfsetfillcolor{currentfill}%
\pgfsetlinewidth{0.501875pt}%
\definecolor{currentstroke}{rgb}{0.000000,0.000000,0.000000}%
\pgfsetstrokecolor{currentstroke}%
\pgfsetdash{}{0pt}%
\pgfsys@defobject{currentmarker}{\pgfqpoint{-0.055556in}{0.000000in}}{\pgfqpoint{0.000000in}{0.000000in}}{%
\pgfpathmoveto{\pgfqpoint{0.000000in}{0.000000in}}%
\pgfpathlineto{\pgfqpoint{-0.055556in}{0.000000in}}%
\pgfusepath{stroke,fill}%
}%
\begin{pgfscope}%
\pgfsys@transformshift{7.200000in}{1.560000in}%
\pgfsys@useobject{currentmarker}{}%
\end{pgfscope}%
\end{pgfscope}%
\begin{pgfscope}%
\pgftext[x=0.944444in,y=1.560000in,right,]{{\sffamily\fontsize{12.000000}{14.400000}\selectfont 5}}%
\end{pgfscope}%
\begin{pgfscope}%
\pgfsetbuttcap%
\pgfsetroundjoin%
\definecolor{currentfill}{rgb}{0.000000,0.000000,0.000000}%
\pgfsetfillcolor{currentfill}%
\pgfsetlinewidth{0.501875pt}%
\definecolor{currentstroke}{rgb}{0.000000,0.000000,0.000000}%
\pgfsetstrokecolor{currentstroke}%
\pgfsetdash{}{0pt}%
\pgfsys@defobject{currentmarker}{\pgfqpoint{0.000000in}{0.000000in}}{\pgfqpoint{0.055556in}{0.000000in}}{%
\pgfpathmoveto{\pgfqpoint{0.000000in}{0.000000in}}%
\pgfpathlineto{\pgfqpoint{0.055556in}{0.000000in}}%
\pgfusepath{stroke,fill}%
}%
\begin{pgfscope}%
\pgfsys@transformshift{1.000000in}{2.520000in}%
\pgfsys@useobject{currentmarker}{}%
\end{pgfscope}%
\end{pgfscope}%
\begin{pgfscope}%
\pgfsetbuttcap%
\pgfsetroundjoin%
\definecolor{currentfill}{rgb}{0.000000,0.000000,0.000000}%
\pgfsetfillcolor{currentfill}%
\pgfsetlinewidth{0.501875pt}%
\definecolor{currentstroke}{rgb}{0.000000,0.000000,0.000000}%
\pgfsetstrokecolor{currentstroke}%
\pgfsetdash{}{0pt}%
\pgfsys@defobject{currentmarker}{\pgfqpoint{-0.055556in}{0.000000in}}{\pgfqpoint{0.000000in}{0.000000in}}{%
\pgfpathmoveto{\pgfqpoint{0.000000in}{0.000000in}}%
\pgfpathlineto{\pgfqpoint{-0.055556in}{0.000000in}}%
\pgfusepath{stroke,fill}%
}%
\begin{pgfscope}%
\pgfsys@transformshift{7.200000in}{2.520000in}%
\pgfsys@useobject{currentmarker}{}%
\end{pgfscope}%
\end{pgfscope}%
\begin{pgfscope}%
\pgftext[x=0.944444in,y=2.520000in,right,]{{\sffamily\fontsize{12.000000}{14.400000}\selectfont 10}}%
\end{pgfscope}%
\begin{pgfscope}%
\pgfsetbuttcap%
\pgfsetroundjoin%
\definecolor{currentfill}{rgb}{0.000000,0.000000,0.000000}%
\pgfsetfillcolor{currentfill}%
\pgfsetlinewidth{0.501875pt}%
\definecolor{currentstroke}{rgb}{0.000000,0.000000,0.000000}%
\pgfsetstrokecolor{currentstroke}%
\pgfsetdash{}{0pt}%
\pgfsys@defobject{currentmarker}{\pgfqpoint{0.000000in}{0.000000in}}{\pgfqpoint{0.055556in}{0.000000in}}{%
\pgfpathmoveto{\pgfqpoint{0.000000in}{0.000000in}}%
\pgfpathlineto{\pgfqpoint{0.055556in}{0.000000in}}%
\pgfusepath{stroke,fill}%
}%
\begin{pgfscope}%
\pgfsys@transformshift{1.000000in}{3.480000in}%
\pgfsys@useobject{currentmarker}{}%
\end{pgfscope}%
\end{pgfscope}%
\begin{pgfscope}%
\pgfsetbuttcap%
\pgfsetroundjoin%
\definecolor{currentfill}{rgb}{0.000000,0.000000,0.000000}%
\pgfsetfillcolor{currentfill}%
\pgfsetlinewidth{0.501875pt}%
\definecolor{currentstroke}{rgb}{0.000000,0.000000,0.000000}%
\pgfsetstrokecolor{currentstroke}%
\pgfsetdash{}{0pt}%
\pgfsys@defobject{currentmarker}{\pgfqpoint{-0.055556in}{0.000000in}}{\pgfqpoint{0.000000in}{0.000000in}}{%
\pgfpathmoveto{\pgfqpoint{0.000000in}{0.000000in}}%
\pgfpathlineto{\pgfqpoint{-0.055556in}{0.000000in}}%
\pgfusepath{stroke,fill}%
}%
\begin{pgfscope}%
\pgfsys@transformshift{7.200000in}{3.480000in}%
\pgfsys@useobject{currentmarker}{}%
\end{pgfscope}%
\end{pgfscope}%
\begin{pgfscope}%
\pgftext[x=0.944444in,y=3.480000in,right,]{{\sffamily\fontsize{12.000000}{14.400000}\selectfont 15}}%
\end{pgfscope}%
\begin{pgfscope}%
\pgfsetbuttcap%
\pgfsetroundjoin%
\definecolor{currentfill}{rgb}{0.000000,0.000000,0.000000}%
\pgfsetfillcolor{currentfill}%
\pgfsetlinewidth{0.501875pt}%
\definecolor{currentstroke}{rgb}{0.000000,0.000000,0.000000}%
\pgfsetstrokecolor{currentstroke}%
\pgfsetdash{}{0pt}%
\pgfsys@defobject{currentmarker}{\pgfqpoint{0.000000in}{0.000000in}}{\pgfqpoint{0.055556in}{0.000000in}}{%
\pgfpathmoveto{\pgfqpoint{0.000000in}{0.000000in}}%
\pgfpathlineto{\pgfqpoint{0.055556in}{0.000000in}}%
\pgfusepath{stroke,fill}%
}%
\begin{pgfscope}%
\pgfsys@transformshift{1.000000in}{4.440000in}%
\pgfsys@useobject{currentmarker}{}%
\end{pgfscope}%
\end{pgfscope}%
\begin{pgfscope}%
\pgfsetbuttcap%
\pgfsetroundjoin%
\definecolor{currentfill}{rgb}{0.000000,0.000000,0.000000}%
\pgfsetfillcolor{currentfill}%
\pgfsetlinewidth{0.501875pt}%
\definecolor{currentstroke}{rgb}{0.000000,0.000000,0.000000}%
\pgfsetstrokecolor{currentstroke}%
\pgfsetdash{}{0pt}%
\pgfsys@defobject{currentmarker}{\pgfqpoint{-0.055556in}{0.000000in}}{\pgfqpoint{0.000000in}{0.000000in}}{%
\pgfpathmoveto{\pgfqpoint{0.000000in}{0.000000in}}%
\pgfpathlineto{\pgfqpoint{-0.055556in}{0.000000in}}%
\pgfusepath{stroke,fill}%
}%
\begin{pgfscope}%
\pgfsys@transformshift{7.200000in}{4.440000in}%
\pgfsys@useobject{currentmarker}{}%
\end{pgfscope}%
\end{pgfscope}%
\begin{pgfscope}%
\pgftext[x=0.944444in,y=4.440000in,right,]{{\sffamily\fontsize{12.000000}{14.400000}\selectfont 20}}%
\end{pgfscope}%
\begin{pgfscope}%
\pgfsetbuttcap%
\pgfsetroundjoin%
\definecolor{currentfill}{rgb}{0.000000,0.000000,0.000000}%
\pgfsetfillcolor{currentfill}%
\pgfsetlinewidth{0.501875pt}%
\definecolor{currentstroke}{rgb}{0.000000,0.000000,0.000000}%
\pgfsetstrokecolor{currentstroke}%
\pgfsetdash{}{0pt}%
\pgfsys@defobject{currentmarker}{\pgfqpoint{0.000000in}{0.000000in}}{\pgfqpoint{0.055556in}{0.000000in}}{%
\pgfpathmoveto{\pgfqpoint{0.000000in}{0.000000in}}%
\pgfpathlineto{\pgfqpoint{0.055556in}{0.000000in}}%
\pgfusepath{stroke,fill}%
}%
\begin{pgfscope}%
\pgfsys@transformshift{1.000000in}{5.400000in}%
\pgfsys@useobject{currentmarker}{}%
\end{pgfscope}%
\end{pgfscope}%
\begin{pgfscope}%
\pgfsetbuttcap%
\pgfsetroundjoin%
\definecolor{currentfill}{rgb}{0.000000,0.000000,0.000000}%
\pgfsetfillcolor{currentfill}%
\pgfsetlinewidth{0.501875pt}%
\definecolor{currentstroke}{rgb}{0.000000,0.000000,0.000000}%
\pgfsetstrokecolor{currentstroke}%
\pgfsetdash{}{0pt}%
\pgfsys@defobject{currentmarker}{\pgfqpoint{-0.055556in}{0.000000in}}{\pgfqpoint{0.000000in}{0.000000in}}{%
\pgfpathmoveto{\pgfqpoint{0.000000in}{0.000000in}}%
\pgfpathlineto{\pgfqpoint{-0.055556in}{0.000000in}}%
\pgfusepath{stroke,fill}%
}%
\begin{pgfscope}%
\pgfsys@transformshift{7.200000in}{5.400000in}%
\pgfsys@useobject{currentmarker}{}%
\end{pgfscope}%
\end{pgfscope}%
\begin{pgfscope}%
\pgftext[x=0.944444in,y=5.400000in,right,]{{\sffamily\fontsize{12.000000}{14.400000}\selectfont 25}}%
\end{pgfscope}%
\begin{pgfscope}%
\pgftext[x=0.662923in,y=3.000000in,,bottom,rotate=90.000000]{{\sffamily\fontsize{22.000000}{26.400000}\selectfont Time, sec}}%
\end{pgfscope}%
\begin{pgfscope}%
\pgfsetbuttcap%
\pgfsetmiterjoin%
\definecolor{currentfill}{rgb}{1.000000,1.000000,1.000000}%
\pgfsetfillcolor{currentfill}%
\pgfsetlinewidth{1.003750pt}%
\definecolor{currentstroke}{rgb}{0.000000,0.000000,0.000000}%
\pgfsetstrokecolor{currentstroke}%
\pgfsetdash{}{0pt}%
\pgfpathmoveto{\pgfqpoint{1.100000in}{4.359336in}}%
\pgfpathlineto{\pgfqpoint{2.612930in}{4.359336in}}%
\pgfpathlineto{\pgfqpoint{2.612930in}{5.300000in}}%
\pgfpathlineto{\pgfqpoint{1.100000in}{5.300000in}}%
\pgfpathclose%
\pgfusepath{stroke,fill}%
\end{pgfscope}%
\begin{pgfscope}%
\pgfsetrectcap%
\pgfsetroundjoin%
\pgfsetlinewidth{1.003750pt}%
\definecolor{currentstroke}{rgb}{0.000000,0.000000,1.000000}%
\pgfsetstrokecolor{currentstroke}%
\pgfsetdash{}{0pt}%
\pgfpathmoveto{\pgfqpoint{1.240000in}{5.138047in}}%
\pgfpathlineto{\pgfqpoint{1.520000in}{5.138047in}}%
\pgfusepath{stroke}%
\end{pgfscope}%
\begin{pgfscope}%
\pgftext[x=1.740000in,y=5.068047in,left,base]{{\sffamily\fontsize{14.400000}{17.280000}\selectfont \(\displaystyle nnz = 10\)}}%
\end{pgfscope}%
\begin{pgfscope}%
\pgfsetrectcap%
\pgfsetroundjoin%
\pgfsetlinewidth{1.003750pt}%
\definecolor{currentstroke}{rgb}{0.000000,0.500000,0.000000}%
\pgfsetstrokecolor{currentstroke}%
\pgfsetdash{}{0pt}%
\pgfpathmoveto{\pgfqpoint{1.240000in}{4.844492in}}%
\pgfpathlineto{\pgfqpoint{1.520000in}{4.844492in}}%
\pgfusepath{stroke}%
\end{pgfscope}%
\begin{pgfscope}%
\pgftext[x=1.740000in,y=4.774492in,left,base]{{\sffamily\fontsize{14.400000}{17.280000}\selectfont \(\displaystyle nnz = 20\)}}%
\end{pgfscope}%
\begin{pgfscope}%
\pgfsetrectcap%
\pgfsetroundjoin%
\pgfsetlinewidth{1.003750pt}%
\definecolor{currentstroke}{rgb}{1.000000,0.000000,0.000000}%
\pgfsetstrokecolor{currentstroke}%
\pgfsetdash{}{0pt}%
\pgfpathmoveto{\pgfqpoint{1.240000in}{4.550938in}}%
\pgfpathlineto{\pgfqpoint{1.520000in}{4.550938in}}%
\pgfusepath{stroke}%
\end{pgfscope}%
\begin{pgfscope}%
\pgftext[x=1.740000in,y=4.480938in,left,base]{{\sffamily\fontsize{14.400000}{17.280000}\selectfont \(\displaystyle nnz = 30\)}}%
\end{pgfscope}%
\end{pgfpicture}%
\makeatother%
\endgroup%

%% file: mem_vs_n.pgf
\begingroup%
\makeatletter%
\begin{pgfpicture}%
\pgfpathrectangle{\pgfpointorigin}{\pgfqpoint{8.000000in}{6.000000in}}%
\pgfusepath{use as bounding box, clip}%
\begin{pgfscope}%
\pgfsetbuttcap%
\pgfsetmiterjoin%
\definecolor{currentfill}{rgb}{1.000000,1.000000,1.000000}%
\pgfsetfillcolor{currentfill}%
\pgfsetlinewidth{0.000000pt}%
\definecolor{currentstroke}{rgb}{1.000000,1.000000,1.000000}%
\pgfsetstrokecolor{currentstroke}%
\pgfsetdash{}{0pt}%
\pgfpathmoveto{\pgfqpoint{0.000000in}{0.000000in}}%
\pgfpathlineto{\pgfqpoint{8.000000in}{0.000000in}}%
\pgfpathlineto{\pgfqpoint{8.000000in}{6.000000in}}%
\pgfpathlineto{\pgfqpoint{0.000000in}{6.000000in}}%
\pgfpathclose%
\pgfusepath{fill}%
\end{pgfscope}%
\begin{pgfscope}%
\pgfsetbuttcap%
\pgfsetmiterjoin%
\definecolor{currentfill}{rgb}{1.000000,1.000000,1.000000}%
\pgfsetfillcolor{currentfill}%
\pgfsetlinewidth{0.000000pt}%
\definecolor{currentstroke}{rgb}{0.000000,0.000000,0.000000}%
\pgfsetstrokecolor{currentstroke}%
\pgfsetstrokeopacity{0.000000}%
\pgfsetdash{}{0pt}%
\pgfpathmoveto{\pgfqpoint{1.000000in}{0.600000in}}%
\pgfpathlineto{\pgfqpoint{7.200000in}{0.600000in}}%
\pgfpathlineto{\pgfqpoint{7.200000in}{5.400000in}}%
\pgfpathlineto{\pgfqpoint{1.000000in}{5.400000in}}%
\pgfpathclose%
\pgfusepath{fill}%
\end{pgfscope}%
\begin{pgfscope}%
\pgfpathrectangle{\pgfqpoint{1.000000in}{0.600000in}}{\pgfqpoint{6.200000in}{4.800000in}} %
\pgfusepath{clip}%
\pgfsetrectcap%
\pgfsetroundjoin%
\pgfsetlinewidth{1.003750pt}%
\definecolor{currentstroke}{rgb}{0.000000,0.000000,1.000000}%
\pgfsetstrokecolor{currentstroke}%
\pgfsetdash{}{0pt}%
\pgfpathmoveto{\pgfqpoint{1.000000in}{0.956706in}}%
\pgfpathlineto{\pgfqpoint{1.387500in}{1.140050in}}%
\pgfpathlineto{\pgfqpoint{1.775000in}{1.348835in}}%
\pgfpathlineto{\pgfqpoint{2.162500in}{1.577032in}}%
\pgfpathlineto{\pgfqpoint{2.550000in}{1.819710in}}%
\pgfpathlineto{\pgfqpoint{2.937500in}{2.076862in}}%
\pgfpathlineto{\pgfqpoint{3.325000in}{2.341994in}}%
\pgfpathlineto{\pgfqpoint{3.712500in}{2.615639in}}%
\pgfpathlineto{\pgfqpoint{4.100000in}{2.894900in}}%
\pgfpathlineto{\pgfqpoint{4.487500in}{3.181368in}}%
\pgfpathlineto{\pgfqpoint{4.875000in}{3.469437in}}%
\pgfpathlineto{\pgfqpoint{5.262500in}{3.764108in}}%
\pgfpathlineto{\pgfqpoint{5.650000in}{4.064342in}}%
\pgfpathlineto{\pgfqpoint{6.037500in}{4.362929in}}%
\pgfpathlineto{\pgfqpoint{6.425000in}{4.665203in}}%
\pgfpathlineto{\pgfqpoint{6.812500in}{4.969504in}}%
\pgfpathlineto{\pgfqpoint{7.200000in}{5.276900in}}%
\pgfusepath{stroke}%
\end{pgfscope}%
\begin{pgfscope}%
\pgfpathrectangle{\pgfqpoint{1.000000in}{0.600000in}}{\pgfqpoint{6.200000in}{4.800000in}} %
\pgfusepath{clip}%
\pgfsetrectcap%
\pgfsetroundjoin%
\pgfsetlinewidth{1.003750pt}%
\definecolor{currentstroke}{rgb}{0.000000,0.500000,0.000000}%
\pgfsetstrokecolor{currentstroke}%
\pgfsetdash{}{0pt}%
\pgfpathmoveto{\pgfqpoint{1.000000in}{0.893652in}}%
\pgfpathlineto{\pgfqpoint{1.387500in}{1.011854in}}%
\pgfpathlineto{\pgfqpoint{1.775000in}{1.138367in}}%
\pgfpathlineto{\pgfqpoint{2.162500in}{1.268830in}}%
\pgfpathlineto{\pgfqpoint{2.550000in}{1.402402in}}%
\pgfpathlineto{\pgfqpoint{2.937500in}{1.536716in}}%
\pgfpathlineto{\pgfqpoint{3.325000in}{1.675045in}}%
\pgfpathlineto{\pgfqpoint{3.712500in}{1.812377in}}%
\pgfpathlineto{\pgfqpoint{4.100000in}{1.953243in}}%
\pgfpathlineto{\pgfqpoint{4.487500in}{2.092401in}}%
\pgfpathlineto{\pgfqpoint{4.875000in}{2.232885in}}%
\pgfpathlineto{\pgfqpoint{5.262500in}{2.374672in}}%
\pgfpathlineto{\pgfqpoint{5.650000in}{2.516774in}}%
\pgfpathlineto{\pgfqpoint{6.037500in}{2.658238in}}%
\pgfpathlineto{\pgfqpoint{6.425000in}{2.800590in}}%
\pgfpathlineto{\pgfqpoint{6.812500in}{2.941870in}}%
\pgfpathlineto{\pgfqpoint{7.200000in}{3.085692in}}%
\pgfusepath{stroke}%
\end{pgfscope}%
\begin{pgfscope}%
\pgfpathrectangle{\pgfqpoint{1.000000in}{0.600000in}}{\pgfqpoint{6.200000in}{4.800000in}} %
\pgfusepath{clip}%
\pgfsetrectcap%
\pgfsetroundjoin%
\pgfsetlinewidth{1.003750pt}%
\definecolor{currentstroke}{rgb}{1.000000,0.000000,0.000000}%
\pgfsetstrokecolor{currentstroke}%
\pgfsetdash{}{0pt}%
\pgfpathmoveto{\pgfqpoint{1.000000in}{0.840487in}}%
\pgfpathlineto{\pgfqpoint{1.387500in}{0.940741in}}%
\pgfpathlineto{\pgfqpoint{1.775000in}{1.020191in}}%
\pgfpathlineto{\pgfqpoint{2.162500in}{1.070075in}}%
\pgfpathlineto{\pgfqpoint{2.550000in}{1.097880in}}%
\pgfpathlineto{\pgfqpoint{2.937500in}{1.124587in}}%
\pgfpathlineto{\pgfqpoint{3.325000in}{1.160667in}}%
\pgfpathlineto{\pgfqpoint{3.712500in}{1.111891in}}%
\pgfpathlineto{\pgfqpoint{4.100000in}{1.143096in}}%
\pgfpathlineto{\pgfqpoint{4.487500in}{1.114556in}}%
\pgfpathlineto{\pgfqpoint{4.875000in}{1.151790in}}%
\pgfpathlineto{\pgfqpoint{5.262500in}{1.148155in}}%
\pgfpathlineto{\pgfqpoint{5.650000in}{1.144264in}}%
\pgfpathlineto{\pgfqpoint{6.037500in}{1.176729in}}%
\pgfpathlineto{\pgfqpoint{6.425000in}{1.188535in}}%
\pgfpathlineto{\pgfqpoint{6.812500in}{1.209926in}}%
\pgfpathlineto{\pgfqpoint{7.200000in}{1.205121in}}%
\pgfusepath{stroke}%
\end{pgfscope}%
\begin{pgfscope}%
\pgfsetrectcap%
\pgfsetmiterjoin%
\pgfsetlinewidth{1.003750pt}%
\definecolor{currentstroke}{rgb}{0.000000,0.000000,0.000000}%
\pgfsetstrokecolor{currentstroke}%
\pgfsetdash{}{0pt}%
\pgfpathmoveto{\pgfqpoint{1.000000in}{5.400000in}}%
\pgfpathlineto{\pgfqpoint{7.200000in}{5.400000in}}%
\pgfusepath{stroke}%
\end{pgfscope}%
\begin{pgfscope}%
\pgfsetrectcap%
\pgfsetmiterjoin%
\pgfsetlinewidth{1.003750pt}%
\definecolor{currentstroke}{rgb}{0.000000,0.000000,0.000000}%
\pgfsetstrokecolor{currentstroke}%
\pgfsetdash{}{0pt}%
\pgfpathmoveto{\pgfqpoint{7.200000in}{0.600000in}}%
\pgfpathlineto{\pgfqpoint{7.200000in}{5.400000in}}%
\pgfusepath{stroke}%
\end{pgfscope}%
\begin{pgfscope}%
\pgfsetrectcap%
\pgfsetmiterjoin%
\pgfsetlinewidth{1.003750pt}%
\definecolor{currentstroke}{rgb}{0.000000,0.000000,0.000000}%
\pgfsetstrokecolor{currentstroke}%
\pgfsetdash{}{0pt}%
\pgfpathmoveto{\pgfqpoint{1.000000in}{0.600000in}}%
\pgfpathlineto{\pgfqpoint{7.200000in}{0.600000in}}%
\pgfusepath{stroke}%
\end{pgfscope}%
\begin{pgfscope}%
\pgfsetrectcap%
\pgfsetmiterjoin%
\pgfsetlinewidth{1.003750pt}%
\definecolor{currentstroke}{rgb}{0.000000,0.000000,0.000000}%
\pgfsetstrokecolor{currentstroke}%
\pgfsetdash{}{0pt}%
\pgfpathmoveto{\pgfqpoint{1.000000in}{0.600000in}}%
\pgfpathlineto{\pgfqpoint{1.000000in}{5.400000in}}%
\pgfusepath{stroke}%
\end{pgfscope}%
\begin{pgfscope}%
\pgfsetbuttcap%
\pgfsetroundjoin%
\definecolor{currentfill}{rgb}{0.000000,0.000000,0.000000}%
\pgfsetfillcolor{currentfill}%
\pgfsetlinewidth{0.501875pt}%
\definecolor{currentstroke}{rgb}{0.000000,0.000000,0.000000}%
\pgfsetstrokecolor{currentstroke}%
\pgfsetdash{}{0pt}%
\pgfsys@defobject{currentmarker}{\pgfqpoint{0.000000in}{0.000000in}}{\pgfqpoint{0.000000in}{0.055556in}}{%
\pgfpathmoveto{\pgfqpoint{0.000000in}{0.000000in}}%
\pgfpathlineto{\pgfqpoint{0.000000in}{0.055556in}}%
\pgfusepath{stroke,fill}%
}%
\begin{pgfscope}%
\pgfsys@transformshift{1.000000in}{0.600000in}%
\pgfsys@useobject{currentmarker}{}%
\end{pgfscope}%
\end{pgfscope}%
\begin{pgfscope}%
\pgfsetbuttcap%
\pgfsetroundjoin%
\definecolor{currentfill}{rgb}{0.000000,0.000000,0.000000}%
\pgfsetfillcolor{currentfill}%
\pgfsetlinewidth{0.501875pt}%
\definecolor{currentstroke}{rgb}{0.000000,0.000000,0.000000}%
\pgfsetstrokecolor{currentstroke}%
\pgfsetdash{}{0pt}%
\pgfsys@defobject{currentmarker}{\pgfqpoint{0.000000in}{-0.055556in}}{\pgfqpoint{0.000000in}{0.000000in}}{%
\pgfpathmoveto{\pgfqpoint{0.000000in}{0.000000in}}%
\pgfpathlineto{\pgfqpoint{0.000000in}{-0.055556in}}%
\pgfusepath{stroke,fill}%
}%
\begin{pgfscope}%
\pgfsys@transformshift{1.000000in}{5.400000in}%
\pgfsys@useobject{currentmarker}{}%
\end{pgfscope}%
\end{pgfscope}%
\begin{pgfscope}%
\pgftext[x=1.000000in,y=0.544444in,,top]{{\sffamily\fontsize{12.000000}{14.400000}\selectfont 1000}}%
\end{pgfscope}%
\begin{pgfscope}%
\pgfsetbuttcap%
\pgfsetroundjoin%
\definecolor{currentfill}{rgb}{0.000000,0.000000,0.000000}%
\pgfsetfillcolor{currentfill}%
\pgfsetlinewidth{0.501875pt}%
\definecolor{currentstroke}{rgb}{0.000000,0.000000,0.000000}%
\pgfsetstrokecolor{currentstroke}%
\pgfsetdash{}{0pt}%
\pgfsys@defobject{currentmarker}{\pgfqpoint{0.000000in}{0.000000in}}{\pgfqpoint{0.000000in}{0.055556in}}{%
\pgfpathmoveto{\pgfqpoint{0.000000in}{0.000000in}}%
\pgfpathlineto{\pgfqpoint{0.000000in}{0.055556in}}%
\pgfusepath{stroke,fill}%
}%
\begin{pgfscope}%
\pgfsys@transformshift{1.775000in}{0.600000in}%
\pgfsys@useobject{currentmarker}{}%
\end{pgfscope}%
\end{pgfscope}%
\begin{pgfscope}%
\pgfsetbuttcap%
\pgfsetroundjoin%
\definecolor{currentfill}{rgb}{0.000000,0.000000,0.000000}%
\pgfsetfillcolor{currentfill}%
\pgfsetlinewidth{0.501875pt}%
\definecolor{currentstroke}{rgb}{0.000000,0.000000,0.000000}%
\pgfsetstrokecolor{currentstroke}%
\pgfsetdash{}{0pt}%
\pgfsys@defobject{currentmarker}{\pgfqpoint{0.000000in}{-0.055556in}}{\pgfqpoint{0.000000in}{0.000000in}}{%
\pgfpathmoveto{\pgfqpoint{0.000000in}{0.000000in}}%
\pgfpathlineto{\pgfqpoint{0.000000in}{-0.055556in}}%
\pgfusepath{stroke,fill}%
}%
\begin{pgfscope}%
\pgfsys@transformshift{1.775000in}{5.400000in}%
\pgfsys@useobject{currentmarker}{}%
\end{pgfscope}%
\end{pgfscope}%
\begin{pgfscope}%
\pgftext[x=1.775000in,y=0.544444in,,top]{{\sffamily\fontsize{12.000000}{14.400000}\selectfont 1500}}%
\end{pgfscope}%
\begin{pgfscope}%
\pgfsetbuttcap%
\pgfsetroundjoin%
\definecolor{currentfill}{rgb}{0.000000,0.000000,0.000000}%
\pgfsetfillcolor{currentfill}%
\pgfsetlinewidth{0.501875pt}%
\definecolor{currentstroke}{rgb}{0.000000,0.000000,0.000000}%
\pgfsetstrokecolor{currentstroke}%
\pgfsetdash{}{0pt}%
\pgfsys@defobject{currentmarker}{\pgfqpoint{0.000000in}{0.000000in}}{\pgfqpoint{0.000000in}{0.055556in}}{%
\pgfpathmoveto{\pgfqpoint{0.000000in}{0.000000in}}%
\pgfpathlineto{\pgfqpoint{0.000000in}{0.055556in}}%
\pgfusepath{stroke,fill}%
}%
\begin{pgfscope}%
\pgfsys@transformshift{2.550000in}{0.600000in}%
\pgfsys@useobject{currentmarker}{}%
\end{pgfscope}%
\end{pgfscope}%
\begin{pgfscope}%
\pgfsetbuttcap%
\pgfsetroundjoin%
\definecolor{currentfill}{rgb}{0.000000,0.000000,0.000000}%
\pgfsetfillcolor{currentfill}%
\pgfsetlinewidth{0.501875pt}%
\definecolor{currentstroke}{rgb}{0.000000,0.000000,0.000000}%
\pgfsetstrokecolor{currentstroke}%
\pgfsetdash{}{0pt}%
\pgfsys@defobject{currentmarker}{\pgfqpoint{0.000000in}{-0.055556in}}{\pgfqpoint{0.000000in}{0.000000in}}{%
\pgfpathmoveto{\pgfqpoint{0.000000in}{0.000000in}}%
\pgfpathlineto{\pgfqpoint{0.000000in}{-0.055556in}}%
\pgfusepath{stroke,fill}%
}%
\begin{pgfscope}%
\pgfsys@transformshift{2.550000in}{5.400000in}%
\pgfsys@useobject{currentmarker}{}%
\end{pgfscope}%
\end{pgfscope}%
\begin{pgfscope}%
\pgftext[x=2.550000in,y=0.544444in,,top]{{\sffamily\fontsize{12.000000}{14.400000}\selectfont 2000}}%
\end{pgfscope}%
\begin{pgfscope}%
\pgfsetbuttcap%
\pgfsetroundjoin%
\definecolor{currentfill}{rgb}{0.000000,0.000000,0.000000}%
\pgfsetfillcolor{currentfill}%
\pgfsetlinewidth{0.501875pt}%
\definecolor{currentstroke}{rgb}{0.000000,0.000000,0.000000}%
\pgfsetstrokecolor{currentstroke}%
\pgfsetdash{}{0pt}%
\pgfsys@defobject{currentmarker}{\pgfqpoint{0.000000in}{0.000000in}}{\pgfqpoint{0.000000in}{0.055556in}}{%
\pgfpathmoveto{\pgfqpoint{0.000000in}{0.000000in}}%
\pgfpathlineto{\pgfqpoint{0.000000in}{0.055556in}}%
\pgfusepath{stroke,fill}%
}%
\begin{pgfscope}%
\pgfsys@transformshift{3.325000in}{0.600000in}%
\pgfsys@useobject{currentmarker}{}%
\end{pgfscope}%
\end{pgfscope}%
\begin{pgfscope}%
\pgfsetbuttcap%
\pgfsetroundjoin%
\definecolor{currentfill}{rgb}{0.000000,0.000000,0.000000}%
\pgfsetfillcolor{currentfill}%
\pgfsetlinewidth{0.501875pt}%
\definecolor{currentstroke}{rgb}{0.000000,0.000000,0.000000}%
\pgfsetstrokecolor{currentstroke}%
\pgfsetdash{}{0pt}%
\pgfsys@defobject{currentmarker}{\pgfqpoint{0.000000in}{-0.055556in}}{\pgfqpoint{0.000000in}{0.000000in}}{%
\pgfpathmoveto{\pgfqpoint{0.000000in}{0.000000in}}%
\pgfpathlineto{\pgfqpoint{0.000000in}{-0.055556in}}%
\pgfusepath{stroke,fill}%
}%
\begin{pgfscope}%
\pgfsys@transformshift{3.325000in}{5.400000in}%
\pgfsys@useobject{currentmarker}{}%
\end{pgfscope}%
\end{pgfscope}%
\begin{pgfscope}%
\pgftext[x=3.325000in,y=0.544444in,,top]{{\sffamily\fontsize{12.000000}{14.400000}\selectfont 2500}}%
\end{pgfscope}%
\begin{pgfscope}%
\pgfsetbuttcap%
\pgfsetroundjoin%
\definecolor{currentfill}{rgb}{0.000000,0.000000,0.000000}%
\pgfsetfillcolor{currentfill}%
\pgfsetlinewidth{0.501875pt}%
\definecolor{currentstroke}{rgb}{0.000000,0.000000,0.000000}%
\pgfsetstrokecolor{currentstroke}%
\pgfsetdash{}{0pt}%
\pgfsys@defobject{currentmarker}{\pgfqpoint{0.000000in}{0.000000in}}{\pgfqpoint{0.000000in}{0.055556in}}{%
\pgfpathmoveto{\pgfqpoint{0.000000in}{0.000000in}}%
\pgfpathlineto{\pgfqpoint{0.000000in}{0.055556in}}%
\pgfusepath{stroke,fill}%
}%
\begin{pgfscope}%
\pgfsys@transformshift{4.100000in}{0.600000in}%
\pgfsys@useobject{currentmarker}{}%
\end{pgfscope}%
\end{pgfscope}%
\begin{pgfscope}%
\pgfsetbuttcap%
\pgfsetroundjoin%
\definecolor{currentfill}{rgb}{0.000000,0.000000,0.000000}%
\pgfsetfillcolor{currentfill}%
\pgfsetlinewidth{0.501875pt}%
\definecolor{currentstroke}{rgb}{0.000000,0.000000,0.000000}%
\pgfsetstrokecolor{currentstroke}%
\pgfsetdash{}{0pt}%
\pgfsys@defobject{currentmarker}{\pgfqpoint{0.000000in}{-0.055556in}}{\pgfqpoint{0.000000in}{0.000000in}}{%
\pgfpathmoveto{\pgfqpoint{0.000000in}{0.000000in}}%
\pgfpathlineto{\pgfqpoint{0.000000in}{-0.055556in}}%
\pgfusepath{stroke,fill}%
}%
\begin{pgfscope}%
\pgfsys@transformshift{4.100000in}{5.400000in}%
\pgfsys@useobject{currentmarker}{}%
\end{pgfscope}%
\end{pgfscope}%
\begin{pgfscope}%
\pgftext[x=4.100000in,y=0.544444in,,top]{{\sffamily\fontsize{12.000000}{14.400000}\selectfont 3000}}%
\end{pgfscope}%
\begin{pgfscope}%
\pgfsetbuttcap%
\pgfsetroundjoin%
\definecolor{currentfill}{rgb}{0.000000,0.000000,0.000000}%
\pgfsetfillcolor{currentfill}%
\pgfsetlinewidth{0.501875pt}%
\definecolor{currentstroke}{rgb}{0.000000,0.000000,0.000000}%
\pgfsetstrokecolor{currentstroke}%
\pgfsetdash{}{0pt}%
\pgfsys@defobject{currentmarker}{\pgfqpoint{0.000000in}{0.000000in}}{\pgfqpoint{0.000000in}{0.055556in}}{%
\pgfpathmoveto{\pgfqpoint{0.000000in}{0.000000in}}%
\pgfpathlineto{\pgfqpoint{0.000000in}{0.055556in}}%
\pgfusepath{stroke,fill}%
}%
\begin{pgfscope}%
\pgfsys@transformshift{4.875000in}{0.600000in}%
\pgfsys@useobject{currentmarker}{}%
\end{pgfscope}%
\end{pgfscope}%
\begin{pgfscope}%
\pgfsetbuttcap%
\pgfsetroundjoin%
\definecolor{currentfill}{rgb}{0.000000,0.000000,0.000000}%
\pgfsetfillcolor{currentfill}%
\pgfsetlinewidth{0.501875pt}%
\definecolor{currentstroke}{rgb}{0.000000,0.000000,0.000000}%
\pgfsetstrokecolor{currentstroke}%
\pgfsetdash{}{0pt}%
\pgfsys@defobject{currentmarker}{\pgfqpoint{0.000000in}{-0.055556in}}{\pgfqpoint{0.000000in}{0.000000in}}{%
\pgfpathmoveto{\pgfqpoint{0.000000in}{0.000000in}}%
\pgfpathlineto{\pgfqpoint{0.000000in}{-0.055556in}}%
\pgfusepath{stroke,fill}%
}%
\begin{pgfscope}%
\pgfsys@transformshift{4.875000in}{5.400000in}%
\pgfsys@useobject{currentmarker}{}%
\end{pgfscope}%
\end{pgfscope}%
\begin{pgfscope}%
\pgftext[x=4.875000in,y=0.544444in,,top]{{\sffamily\fontsize{12.000000}{14.400000}\selectfont 3500}}%
\end{pgfscope}%
\begin{pgfscope}%
\pgfsetbuttcap%
\pgfsetroundjoin%
\definecolor{currentfill}{rgb}{0.000000,0.000000,0.000000}%
\pgfsetfillcolor{currentfill}%
\pgfsetlinewidth{0.501875pt}%
\definecolor{currentstroke}{rgb}{0.000000,0.000000,0.000000}%
\pgfsetstrokecolor{currentstroke}%
\pgfsetdash{}{0pt}%
\pgfsys@defobject{currentmarker}{\pgfqpoint{0.000000in}{0.000000in}}{\pgfqpoint{0.000000in}{0.055556in}}{%
\pgfpathmoveto{\pgfqpoint{0.000000in}{0.000000in}}%
\pgfpathlineto{\pgfqpoint{0.000000in}{0.055556in}}%
\pgfusepath{stroke,fill}%
}%
\begin{pgfscope}%
\pgfsys@transformshift{5.650000in}{0.600000in}%
\pgfsys@useobject{currentmarker}{}%
\end{pgfscope}%
\end{pgfscope}%
\begin{pgfscope}%
\pgfsetbuttcap%
\pgfsetroundjoin%
\definecolor{currentfill}{rgb}{0.000000,0.000000,0.000000}%
\pgfsetfillcolor{currentfill}%
\pgfsetlinewidth{0.501875pt}%
\definecolor{currentstroke}{rgb}{0.000000,0.000000,0.000000}%
\pgfsetstrokecolor{currentstroke}%
\pgfsetdash{}{0pt}%
\pgfsys@defobject{currentmarker}{\pgfqpoint{0.000000in}{-0.055556in}}{\pgfqpoint{0.000000in}{0.000000in}}{%
\pgfpathmoveto{\pgfqpoint{0.000000in}{0.000000in}}%
\pgfpathlineto{\pgfqpoint{0.000000in}{-0.055556in}}%
\pgfusepath{stroke,fill}%
}%
\begin{pgfscope}%
\pgfsys@transformshift{5.650000in}{5.400000in}%
\pgfsys@useobject{currentmarker}{}%
\end{pgfscope}%
\end{pgfscope}%
\begin{pgfscope}%
\pgftext[x=5.650000in,y=0.544444in,,top]{{\sffamily\fontsize{12.000000}{14.400000}\selectfont 4000}}%
\end{pgfscope}%
\begin{pgfscope}%
\pgfsetbuttcap%
\pgfsetroundjoin%
\definecolor{currentfill}{rgb}{0.000000,0.000000,0.000000}%
\pgfsetfillcolor{currentfill}%
\pgfsetlinewidth{0.501875pt}%
\definecolor{currentstroke}{rgb}{0.000000,0.000000,0.000000}%
\pgfsetstrokecolor{currentstroke}%
\pgfsetdash{}{0pt}%
\pgfsys@defobject{currentmarker}{\pgfqpoint{0.000000in}{0.000000in}}{\pgfqpoint{0.000000in}{0.055556in}}{%
\pgfpathmoveto{\pgfqpoint{0.000000in}{0.000000in}}%
\pgfpathlineto{\pgfqpoint{0.000000in}{0.055556in}}%
\pgfusepath{stroke,fill}%
}%
\begin{pgfscope}%
\pgfsys@transformshift{6.425000in}{0.600000in}%
\pgfsys@useobject{currentmarker}{}%
\end{pgfscope}%
\end{pgfscope}%
\begin{pgfscope}%
\pgfsetbuttcap%
\pgfsetroundjoin%
\definecolor{currentfill}{rgb}{0.000000,0.000000,0.000000}%
\pgfsetfillcolor{currentfill}%
\pgfsetlinewidth{0.501875pt}%
\definecolor{currentstroke}{rgb}{0.000000,0.000000,0.000000}%
\pgfsetstrokecolor{currentstroke}%
\pgfsetdash{}{0pt}%
\pgfsys@defobject{currentmarker}{\pgfqpoint{0.000000in}{-0.055556in}}{\pgfqpoint{0.000000in}{0.000000in}}{%
\pgfpathmoveto{\pgfqpoint{0.000000in}{0.000000in}}%
\pgfpathlineto{\pgfqpoint{0.000000in}{-0.055556in}}%
\pgfusepath{stroke,fill}%
}%
\begin{pgfscope}%
\pgfsys@transformshift{6.425000in}{5.400000in}%
\pgfsys@useobject{currentmarker}{}%
\end{pgfscope}%
\end{pgfscope}%
\begin{pgfscope}%
\pgftext[x=6.425000in,y=0.544444in,,top]{{\sffamily\fontsize{12.000000}{14.400000}\selectfont 4500}}%
\end{pgfscope}%
\begin{pgfscope}%
\pgfsetbuttcap%
\pgfsetroundjoin%
\definecolor{currentfill}{rgb}{0.000000,0.000000,0.000000}%
\pgfsetfillcolor{currentfill}%
\pgfsetlinewidth{0.501875pt}%
\definecolor{currentstroke}{rgb}{0.000000,0.000000,0.000000}%
\pgfsetstrokecolor{currentstroke}%
\pgfsetdash{}{0pt}%
\pgfsys@defobject{currentmarker}{\pgfqpoint{0.000000in}{0.000000in}}{\pgfqpoint{0.000000in}{0.055556in}}{%
\pgfpathmoveto{\pgfqpoint{0.000000in}{0.000000in}}%
\pgfpathlineto{\pgfqpoint{0.000000in}{0.055556in}}%
\pgfusepath{stroke,fill}%
}%
\begin{pgfscope}%
\pgfsys@transformshift{7.200000in}{0.600000in}%
\pgfsys@useobject{currentmarker}{}%
\end{pgfscope}%
\end{pgfscope}%
\begin{pgfscope}%
\pgfsetbuttcap%
\pgfsetroundjoin%
\definecolor{currentfill}{rgb}{0.000000,0.000000,0.000000}%
\pgfsetfillcolor{currentfill}%
\pgfsetlinewidth{0.501875pt}%
\definecolor{currentstroke}{rgb}{0.000000,0.000000,0.000000}%
\pgfsetstrokecolor{currentstroke}%
\pgfsetdash{}{0pt}%
\pgfsys@defobject{currentmarker}{\pgfqpoint{0.000000in}{-0.055556in}}{\pgfqpoint{0.000000in}{0.000000in}}{%
\pgfpathmoveto{\pgfqpoint{0.000000in}{0.000000in}}%
\pgfpathlineto{\pgfqpoint{0.000000in}{-0.055556in}}%
\pgfusepath{stroke,fill}%
}%
\begin{pgfscope}%
\pgfsys@transformshift{7.200000in}{5.400000in}%
\pgfsys@useobject{currentmarker}{}%
\end{pgfscope}%
\end{pgfscope}%
\begin{pgfscope}%
\pgftext[x=7.200000in,y=0.544444in,,top]{{\sffamily\fontsize{12.000000}{14.400000}\selectfont 5000}}%
\end{pgfscope}%
\begin{pgfscope}%
\pgftext[x=4.100000in,y=0.313705in,,top]{{\sffamily\fontsize{22.000000}{26.400000}\selectfont Number of vertices, \(\displaystyle n\)}}%
\end{pgfscope}%
\begin{pgfscope}%
\pgfsetbuttcap%
\pgfsetroundjoin%
\definecolor{currentfill}{rgb}{0.000000,0.000000,0.000000}%
\pgfsetfillcolor{currentfill}%
\pgfsetlinewidth{0.501875pt}%
\definecolor{currentstroke}{rgb}{0.000000,0.000000,0.000000}%
\pgfsetstrokecolor{currentstroke}%
\pgfsetdash{}{0pt}%
\pgfsys@defobject{currentmarker}{\pgfqpoint{0.000000in}{0.000000in}}{\pgfqpoint{0.055556in}{0.000000in}}{%
\pgfpathmoveto{\pgfqpoint{0.000000in}{0.000000in}}%
\pgfpathlineto{\pgfqpoint{0.055556in}{0.000000in}}%
\pgfusepath{stroke,fill}%
}%
\begin{pgfscope}%
\pgfsys@transformshift{1.000000in}{0.600000in}%
\pgfsys@useobject{currentmarker}{}%
\end{pgfscope}%
\end{pgfscope}%
\begin{pgfscope}%
\pgfsetbuttcap%
\pgfsetroundjoin%
\definecolor{currentfill}{rgb}{0.000000,0.000000,0.000000}%
\pgfsetfillcolor{currentfill}%
\pgfsetlinewidth{0.501875pt}%
\definecolor{currentstroke}{rgb}{0.000000,0.000000,0.000000}%
\pgfsetstrokecolor{currentstroke}%
\pgfsetdash{}{0pt}%
\pgfsys@defobject{currentmarker}{\pgfqpoint{-0.055556in}{0.000000in}}{\pgfqpoint{0.000000in}{0.000000in}}{%
\pgfpathmoveto{\pgfqpoint{0.000000in}{0.000000in}}%
\pgfpathlineto{\pgfqpoint{-0.055556in}{0.000000in}}%
\pgfusepath{stroke,fill}%
}%
\begin{pgfscope}%
\pgfsys@transformshift{7.200000in}{0.600000in}%
\pgfsys@useobject{currentmarker}{}%
\end{pgfscope}%
\end{pgfscope}%
\begin{pgfscope}%
\pgftext[x=0.944444in,y=0.600000in,right,]{{\sffamily\fontsize{12.000000}{14.400000}\selectfont 0}}%
\end{pgfscope}%
\begin{pgfscope}%
\pgfsetbuttcap%
\pgfsetroundjoin%
\definecolor{currentfill}{rgb}{0.000000,0.000000,0.000000}%
\pgfsetfillcolor{currentfill}%
\pgfsetlinewidth{0.501875pt}%
\definecolor{currentstroke}{rgb}{0.000000,0.000000,0.000000}%
\pgfsetstrokecolor{currentstroke}%
\pgfsetdash{}{0pt}%
\pgfsys@defobject{currentmarker}{\pgfqpoint{0.000000in}{0.000000in}}{\pgfqpoint{0.055556in}{0.000000in}}{%
\pgfpathmoveto{\pgfqpoint{0.000000in}{0.000000in}}%
\pgfpathlineto{\pgfqpoint{0.055556in}{0.000000in}}%
\pgfusepath{stroke,fill}%
}%
\begin{pgfscope}%
\pgfsys@transformshift{1.000000in}{1.800000in}%
\pgfsys@useobject{currentmarker}{}%
\end{pgfscope}%
\end{pgfscope}%
\begin{pgfscope}%
\pgfsetbuttcap%
\pgfsetroundjoin%
\definecolor{currentfill}{rgb}{0.000000,0.000000,0.000000}%
\pgfsetfillcolor{currentfill}%
\pgfsetlinewidth{0.501875pt}%
\definecolor{currentstroke}{rgb}{0.000000,0.000000,0.000000}%
\pgfsetstrokecolor{currentstroke}%
\pgfsetdash{}{0pt}%
\pgfsys@defobject{currentmarker}{\pgfqpoint{-0.055556in}{0.000000in}}{\pgfqpoint{0.000000in}{0.000000in}}{%
\pgfpathmoveto{\pgfqpoint{0.000000in}{0.000000in}}%
\pgfpathlineto{\pgfqpoint{-0.055556in}{0.000000in}}%
\pgfusepath{stroke,fill}%
}%
\begin{pgfscope}%
\pgfsys@transformshift{7.200000in}{1.800000in}%
\pgfsys@useobject{currentmarker}{}%
\end{pgfscope}%
\end{pgfscope}%
\begin{pgfscope}%
\pgftext[x=0.944444in,y=1.800000in,right,]{{\sffamily\fontsize{12.000000}{14.400000}\selectfont 50}}%
\end{pgfscope}%
\begin{pgfscope}%
\pgfsetbuttcap%
\pgfsetroundjoin%
\definecolor{currentfill}{rgb}{0.000000,0.000000,0.000000}%
\pgfsetfillcolor{currentfill}%
\pgfsetlinewidth{0.501875pt}%
\definecolor{currentstroke}{rgb}{0.000000,0.000000,0.000000}%
\pgfsetstrokecolor{currentstroke}%
\pgfsetdash{}{0pt}%
\pgfsys@defobject{currentmarker}{\pgfqpoint{0.000000in}{0.000000in}}{\pgfqpoint{0.055556in}{0.000000in}}{%
\pgfpathmoveto{\pgfqpoint{0.000000in}{0.000000in}}%
\pgfpathlineto{\pgfqpoint{0.055556in}{0.000000in}}%
\pgfusepath{stroke,fill}%
}%
\begin{pgfscope}%
\pgfsys@transformshift{1.000000in}{3.000000in}%
\pgfsys@useobject{currentmarker}{}%
\end{pgfscope}%
\end{pgfscope}%
\begin{pgfscope}%
\pgfsetbuttcap%
\pgfsetroundjoin%
\definecolor{currentfill}{rgb}{0.000000,0.000000,0.000000}%
\pgfsetfillcolor{currentfill}%
\pgfsetlinewidth{0.501875pt}%
\definecolor{currentstroke}{rgb}{0.000000,0.000000,0.000000}%
\pgfsetstrokecolor{currentstroke}%
\pgfsetdash{}{0pt}%
\pgfsys@defobject{currentmarker}{\pgfqpoint{-0.055556in}{0.000000in}}{\pgfqpoint{0.000000in}{0.000000in}}{%
\pgfpathmoveto{\pgfqpoint{0.000000in}{0.000000in}}%
\pgfpathlineto{\pgfqpoint{-0.055556in}{0.000000in}}%
\pgfusepath{stroke,fill}%
}%
\begin{pgfscope}%
\pgfsys@transformshift{7.200000in}{3.000000in}%
\pgfsys@useobject{currentmarker}{}%
\end{pgfscope}%
\end{pgfscope}%
\begin{pgfscope}%
\pgftext[x=0.944444in,y=3.000000in,right,]{{\sffamily\fontsize{12.000000}{14.400000}\selectfont 100}}%
\end{pgfscope}%
\begin{pgfscope}%
\pgfsetbuttcap%
\pgfsetroundjoin%
\definecolor{currentfill}{rgb}{0.000000,0.000000,0.000000}%
\pgfsetfillcolor{currentfill}%
\pgfsetlinewidth{0.501875pt}%
\definecolor{currentstroke}{rgb}{0.000000,0.000000,0.000000}%
\pgfsetstrokecolor{currentstroke}%
\pgfsetdash{}{0pt}%
\pgfsys@defobject{currentmarker}{\pgfqpoint{0.000000in}{0.000000in}}{\pgfqpoint{0.055556in}{0.000000in}}{%
\pgfpathmoveto{\pgfqpoint{0.000000in}{0.000000in}}%
\pgfpathlineto{\pgfqpoint{0.055556in}{0.000000in}}%
\pgfusepath{stroke,fill}%
}%
\begin{pgfscope}%
\pgfsys@transformshift{1.000000in}{4.200000in}%
\pgfsys@useobject{currentmarker}{}%
\end{pgfscope}%
\end{pgfscope}%
\begin{pgfscope}%
\pgfsetbuttcap%
\pgfsetroundjoin%
\definecolor{currentfill}{rgb}{0.000000,0.000000,0.000000}%
\pgfsetfillcolor{currentfill}%
\pgfsetlinewidth{0.501875pt}%
\definecolor{currentstroke}{rgb}{0.000000,0.000000,0.000000}%
\pgfsetstrokecolor{currentstroke}%
\pgfsetdash{}{0pt}%
\pgfsys@defobject{currentmarker}{\pgfqpoint{-0.055556in}{0.000000in}}{\pgfqpoint{0.000000in}{0.000000in}}{%
\pgfpathmoveto{\pgfqpoint{0.000000in}{0.000000in}}%
\pgfpathlineto{\pgfqpoint{-0.055556in}{0.000000in}}%
\pgfusepath{stroke,fill}%
}%
\begin{pgfscope}%
\pgfsys@transformshift{7.200000in}{4.200000in}%
\pgfsys@useobject{currentmarker}{}%
\end{pgfscope}%
\end{pgfscope}%
\begin{pgfscope}%
\pgftext[x=0.944444in,y=4.200000in,right,]{{\sffamily\fontsize{12.000000}{14.400000}\selectfont 150}}%
\end{pgfscope}%
\begin{pgfscope}%
\pgfsetbuttcap%
\pgfsetroundjoin%
\definecolor{currentfill}{rgb}{0.000000,0.000000,0.000000}%
\pgfsetfillcolor{currentfill}%
\pgfsetlinewidth{0.501875pt}%
\definecolor{currentstroke}{rgb}{0.000000,0.000000,0.000000}%
\pgfsetstrokecolor{currentstroke}%
\pgfsetdash{}{0pt}%
\pgfsys@defobject{currentmarker}{\pgfqpoint{0.000000in}{0.000000in}}{\pgfqpoint{0.055556in}{0.000000in}}{%
\pgfpathmoveto{\pgfqpoint{0.000000in}{0.000000in}}%
\pgfpathlineto{\pgfqpoint{0.055556in}{0.000000in}}%
\pgfusepath{stroke,fill}%
}%
\begin{pgfscope}%
\pgfsys@transformshift{1.000000in}{5.400000in}%
\pgfsys@useobject{currentmarker}{}%
\end{pgfscope}%
\end{pgfscope}%
\begin{pgfscope}%
\pgfsetbuttcap%
\pgfsetroundjoin%
\definecolor{currentfill}{rgb}{0.000000,0.000000,0.000000}%
\pgfsetfillcolor{currentfill}%
\pgfsetlinewidth{0.501875pt}%
\definecolor{currentstroke}{rgb}{0.000000,0.000000,0.000000}%
\pgfsetstrokecolor{currentstroke}%
\pgfsetdash{}{0pt}%
\pgfsys@defobject{currentmarker}{\pgfqpoint{-0.055556in}{0.000000in}}{\pgfqpoint{0.000000in}{0.000000in}}{%
\pgfpathmoveto{\pgfqpoint{0.000000in}{0.000000in}}%
\pgfpathlineto{\pgfqpoint{-0.055556in}{0.000000in}}%
\pgfusepath{stroke,fill}%
}%
\begin{pgfscope}%
\pgfsys@transformshift{7.200000in}{5.400000in}%
\pgfsys@useobject{currentmarker}{}%
\end{pgfscope}%
\end{pgfscope}%
\begin{pgfscope}%
\pgftext[x=0.944444in,y=5.400000in,right,]{{\sffamily\fontsize{12.000000}{14.400000}\selectfont 200}}%
\end{pgfscope}%
\begin{pgfscope}%
\pgftext[x=0.556885in,y=3.000000in,,bottom,rotate=90.000000]{{\sffamily\fontsize{22.000000}{26.400000}\selectfont Memory, MB}}%
\end{pgfscope}%
\begin{pgfscope}%
\pgfsetbuttcap%
\pgfsetmiterjoin%
\definecolor{currentfill}{rgb}{1.000000,1.000000,1.000000}%
\pgfsetfillcolor{currentfill}%
\pgfsetlinewidth{1.003750pt}%
\definecolor{currentstroke}{rgb}{0.000000,0.000000,0.000000}%
\pgfsetstrokecolor{currentstroke}%
\pgfsetdash{}{0pt}%
\pgfpathmoveto{\pgfqpoint{1.100000in}{4.359336in}}%
\pgfpathlineto{\pgfqpoint{2.612930in}{4.359336in}}%
\pgfpathlineto{\pgfqpoint{2.612930in}{5.300000in}}%
\pgfpathlineto{\pgfqpoint{1.100000in}{5.300000in}}%
\pgfpathclose%
\pgfusepath{stroke,fill}%
\end{pgfscope}%
\begin{pgfscope}%
\pgfsetrectcap%
\pgfsetroundjoin%
\pgfsetlinewidth{1.003750pt}%
\definecolor{currentstroke}{rgb}{0.000000,0.000000,1.000000}%
\pgfsetstrokecolor{currentstroke}%
\pgfsetdash{}{0pt}%
\pgfpathmoveto{\pgfqpoint{1.240000in}{5.138047in}}%
\pgfpathlineto{\pgfqpoint{1.520000in}{5.138047in}}%
\pgfusepath{stroke}%
\end{pgfscope}%
\begin{pgfscope}%
\pgftext[x=1.740000in,y=5.068047in,left,base]{{\sffamily\fontsize{14.400000}{17.280000}\selectfont \(\displaystyle nnz = 30\)}}%
\end{pgfscope}%
\begin{pgfscope}%
\pgfsetrectcap%
\pgfsetroundjoin%
\pgfsetlinewidth{1.003750pt}%
\definecolor{currentstroke}{rgb}{0.000000,0.500000,0.000000}%
\pgfsetstrokecolor{currentstroke}%
\pgfsetdash{}{0pt}%
\pgfpathmoveto{\pgfqpoint{1.240000in}{4.844492in}}%
\pgfpathlineto{\pgfqpoint{1.520000in}{4.844492in}}%
\pgfusepath{stroke}%
\end{pgfscope}%
\begin{pgfscope}%
\pgftext[x=1.740000in,y=4.774492in,left,base]{{\sffamily\fontsize{14.400000}{17.280000}\selectfont \(\displaystyle nnz = 20\)}}%
\end{pgfscope}%
\begin{pgfscope}%
\pgfsetrectcap%
\pgfsetroundjoin%
\pgfsetlinewidth{1.003750pt}%
\definecolor{currentstroke}{rgb}{1.000000,0.000000,0.000000}%
\pgfsetstrokecolor{currentstroke}%
\pgfsetdash{}{0pt}%
\pgfpathmoveto{\pgfqpoint{1.240000in}{4.550938in}}%
\pgfpathlineto{\pgfqpoint{1.520000in}{4.550938in}}%
\pgfusepath{stroke}%
\end{pgfscope}%
\begin{pgfscope}%
\pgftext[x=1.740000in,y=4.480938in,left,base]{{\sffamily\fontsize{14.400000}{17.280000}\selectfont \(\displaystyle nnz = 10\)}}%
\end{pgfscope}%
\end{pgfpicture}%
\makeatother%
\endgroup%